\documentclass[sigconf]{acmart}

\usepackage{booktabs} 

\usepackage{multirow}
\usepackage{subfigure}
\usepackage{color}
\usepackage{url}
\usepackage{amssymb}
\usepackage{pifont}

\begin{document}
\title{Signed Node Relevance Measurements}

\author{Tyler Derr}
\affiliation{%
  \institution{Data Science and Engineering Lab\\Michigan State University}
}
\email{derrtyle@msu.edu}

\author{Chenxing Wang}
\affiliation{%
  \institution{Data Science and Engineering Lab\\Michigan State University}
}
\email{wangch88@msu.edu}

\author{Suhang Wang}
\affiliation{%
  \institution{Computer Science and Engineering\\Arizona State University}
}
\email{suhang.wang@asu.edu}

\author{Jiliang Tang}
\affiliation{%
  \institution{Data Science and Engineering Lab\\Michigan State University}
}
\email{tangjili@msu.edu}

\begin{abstract}
In this paper, we perform the initial and comprehensive study on the problem of measuring node relevance on signed social networks.  We design numerous relevance measurements for signed social networks from both local and global perspectives and investigate the connection between signed relevance measurements, balance theory and signed network properties. Experimental results are conducted to study the effects of signed relevance measurements with four real-world datasets on signed network analysis tasks. 
\end{abstract}

\maketitle

\section{Introduction}

Traditionally network analysis has focused on unsigned networks. However, many online social networking services provide mechanisms that allow users to create not only positive links, but also negative relations. These social networks with both positive and negative links are known as signed social networks, where the negative links users give can denote their foes (e.g., Slashdot), those they distrust (e.g., Epinions), or ``unfriended'' friends and blocked users (e.g., Facebook and Twitter). It is due to this diverse set of signed networks appearing in today's social media that has lead to their increased attention in the recent years; as well as the increased availability due to the more and more popularity of online social media~\cite{yang2007community,leskovec2010signed,hsieh2012low,anchuri2012communities}.
 
Node relevance, which measures how relevant two nodes are in a social network, is one of the keystones of social network analysis. This has been shown by their usage in diverse social network analysis tasks and applications such as link prediction~\cite{backstrom2011supervised,yin2010unified}, node classification~\cite{bhagat2011node}, community detection~\cite{tang2010community}, search and recommendations~\cite{yin2012longtail}. The vast majority of existing node relevance measurements have been designed for unsigned networks ( or social networks with only positive links)~\cite{barabasi1999emergence,adamic2003friends}. However, the availability of negative links in signed networks poses tremendous challenges to unsigned relevance measurements. For instance, most unsigned relevance measurements require all links positive~\cite{scott2012social}. Meanwhile, the fundamental principles and theories of signed networks are substantially different from those of unsigned networks. For example, some social theories such as balance theory~\cite{heider1946attitudes} are only applicable to signed networks, while social theories for unsigned networks such as homophily may not be applicable to signed networks~\cite{tang2014distrust}.  Therefore, relevance measurements for signed networks need dedicated efforts since it cannot be executed by simply applying those for unsigned networks. 

On the other hand, the existence of negative links also brings about unprecedented opportunities in signed relevance measurements. It is evident from recent research that negative links have significant added value over positive links in various analytical tasks. For example, a small number of negative links can significantly improve positive link prediction~\cite{guha2004propagation,leskovec2010predicting}, and they can also boost the performance of recommender systems~\cite{victor2009trust,ma2009learning}. Thereby, negative links could offer the potential to help us develop novel relevance measurements for signed networks. There are a few very recent works in designing node similarities for link prediction~\cite{symeonidis2014transitive,Jung2016srwr}. However, a general and systematic investigation on signed relevance measurements and their effects on signed network analysis are still desired since it can greatly advance our understandings about signed social networks.  

In this paper, we perform the initial and comprehensive study on the problem of measuring node relevance on signed social networks.  Analogous to node relevance research in unsigned networks, we aim to investigate the following: (a) how to make use of both positive and negative links in signed relevance measurements; and (b) what are the effects of these measurements on signed network analysis. The main contributions of the paper are summarized as follows:

\begin{itemize}
	\item Design numerous relevance measurements for signed social networks from both local and global perspectives.;
	\item Investigate the connection between signed relevance measurements, and balance theory and signed network properties; and 
	\item Study the effects of signed relevance measurements with four real-world datasets on two signed network analysis tasks - link prediction and tie strength prediction.
\end{itemize}

The rest of this paper is organized as follows. In Section 2, we review related work in node relevance measurements and signed networks. We describe the four signed network datasets used in this paper, a preliminary analysis of the data, along with some validation for balance theory in Section 3. Then, in Section 4, we present numerous node relevance measurements specific to signed networks. In Section 5 we perform experiments for predicting links and tie strength predictions when using the node relevance algorithms previously discussed in Section 4. Finally, conclusions are given along with our future work in Section 6.

\section{Related Work}

Our work is related to node relevance measurements and signed network analysis. In the following subsections, we will briefly overview them. 

\subsection{Node Relevance Measurements}

Measuring node relevance is fundamental to social network analysis. Most of existing node relevance measurements have been developed for unsigned social networks. According to the used information, we can roughly categorize them into local and global methods.  Local methods, commonly referred as structural equivalence~\cite{lorrain1971structural}, use local node neighborhood information. Representative local measurements include common neighbors and its variants, Jaccard Index and its variants such as Sorensen Index, Adamic-Adar Index~\cite{adamic2003friends}, and Preferential Attachment Index~\cite{barabasi1999emergence}. Global methods not only utilize the local neighborhoods but also propagate the relevance information through the whole network. Representative global measurements include Katz~\cite{katz1953sim}, SimRank~\cite{jeh2002simrank}, ASCOS and ASCOS++~\cite{chen2013ascos,chen2015ascos}, and random walk with restart (RWR) and its variants~\cite{tong2006fast}. One recent work extends RWR for personalized ranking in signed social networks~\cite{Jung2016srwr} and a few recent works studied node similarities for link prediction~\cite{symeonidis2014transitive}. However, to the best of our knowledge, this work is the initial and comprehensive study about node relevance measurements in signed social networks. 

\subsection{Signed Network Analysis}

 With roots in social psychology~\cite{heider1946attitudes,cartwright1956structural}, signed network analysis has attracted increasing attention in recent years. However, the development of tasks of signed social network analysis is highly imbalanced~\cite{tang2016survey}. Some tasks have been extensively studied such as social balance in signed networks~\cite{facchetti2011computing, zheng2015social},  link prediction~\cite{leskovec2010predicting,chiang2011exploiting}, and community detection~\cite{chiang2012scalable,kunegis2010spectral}; some tasks are still in the very early stages of development such as signed network embedding~\cite{wang2017signed} and negative link prediction~\cite{tang2015negative}; while others have not been comprehensively investigated such as node relevance measurements and signed network modeling. A comprehensive overview about signed network analysis can be found in~\cite{tang2016survey}.

\section{ Data Analysis} \label{data_analysis}

In this section, we will first introduce the datasets we will use for this study and then perform preliminary analysis with them.

\subsection{Datasets}

In this work, we collect four signed network datasets to study signed relevance measurements, i.e., Bitcoin-Alpha\footnote{http://www.btcalpha.com}, Bitcoin-OTC\footnote{https://www.bitcoin-otc.com}, Slashdot\footnote{http://www.slashdot.org} and Epinions\footnote{http://www.epinions.com}. Below we describe more details about these datasets.

The Alpha network is a signed network we collected from Bitcoin Alpha. Similarly we collected Bitcoin-OTC from Bitcoin OTC. Both of these datasets were collected from publicly available data from their respective websites. The two Bitcoin sites are open market websites that allow users to buy and sell things. Due to the anonymity behind users' Bitcoin account, users of these websites form trust networks to prevent against scammers (e.g., fake users who are just attempting to have another user send them bitcoins, but never deliver their end of the deal, which is usually the delivery of some other monetary good). In addition to the signed networks,  users in both websites can specify scores in [1,10] (or [-10,-1]) to indicate the positive (or negative) tie strength. All the data from these websites was exhaustively crawled on December 18th of 2016. Note that negative links in both websites are visible to the public. 

The Slashdot dataset was obtained from~\cite{kunegis2009slashdot}. Slashdot focuses on providing technology news since 1997. One of the unique features is that since 2002 the website has allowed users to explicitly mark other users as their friends (positive links) or foes (negative links).  Note that negative links in Slashdot are only visible to users who login the system. 

We have collected a dataset from the product review site Epinions where users can establish trust (positive) and distrust (negative) links. In addition, users can write reviews for items from certain pre-defined categories. We also collect category information for each user. Such information will serve as the ground-truth for the task of node classification. More details will be discussed in the following sections. Note that negative links in Epinions are totally invisible to the public but in the dataset, negative links were given by Epinions staff for the research purpose. 

\begin{table}
\begin{center}
\caption{Statistics of four signed social networks.}
\label{tab:datasets}
\begin{tabular}{c|c|c|c}
	\hline
	Network & \# Users & \# Positive &\# Negative  \\\hline
	Bitcoin-Alpha  & 3,784 & 22,651  & 1,556     \\	
	Bitcoin-OTC & 5,901 & 32,448  & 3,526     \\	
	Slashdot  & 79,116 & 392,179  & 123,218     \\	
	Epinions  & 131,828 & 717,667  & 123,705 \\ \hline
\end{tabular}
\end{center}
\end{table}

Some statistics are demonstrated in Table~\ref{tab:datasets}. We note from the table that in all datasets, negative links are sparser than positive links, thus negative links could have different properties from positive links. Meanwhile, previous studies suggested that balance theory is helpful to explain social phenomena in signed networks~\cite{leskovec2010predicting}. Thus, in the following subsections, we will study properties of negative links analogous to positive links and validate balance theory in four real-world signed networks.

\subsection{Degree Distributions}

As we know, the distributions of in- or out-degrees of positive links in unsigned networks follow power-law distributions -- most nodes with small degrees while a few nodes with large degrees~\cite{barabasi1999emergence}.  In this subsection, we examine whether similar distributions can be observed for positive and negative links in signed social networks. 

For each user, we calculate the numbers of in- and out-degrees for positive and negative links, separately. The distributions of in- and out-degrees of positive and negative links in four signed networks are demonstrated in Figure~\ref{fig:degree-distributions}. From the figure, it is clearly observed that the degree distributions of positive and negative links in all four signed networks also follow power-law distributions. For instance, a few nodes give a large number of negative links; while many nodes only give few negative links. 

\begin{figure}
	\begin{center}
	  \subfigure[Bitcoin-Alpha] {\includegraphics[scale=0.15]{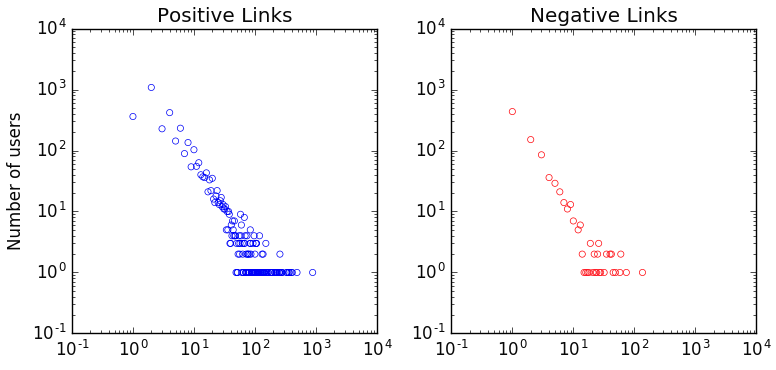}}  
	  \subfigure[Bitcoin-BTC]{\includegraphics[scale=0.15]{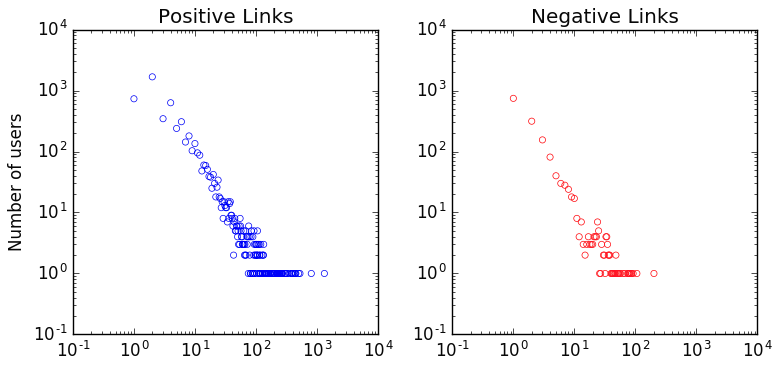}}	
  	  \subfigure[Epinions]{\includegraphics[scale=0.15]{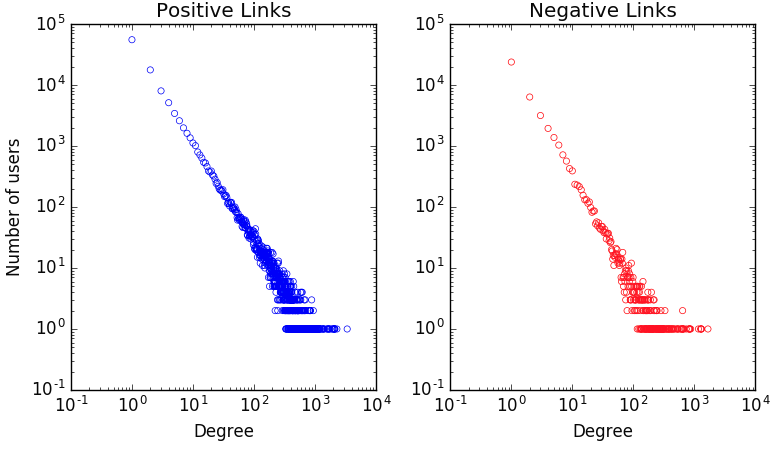}}
	  \subfigure[Slashdot]{\includegraphics[scale=0.15]{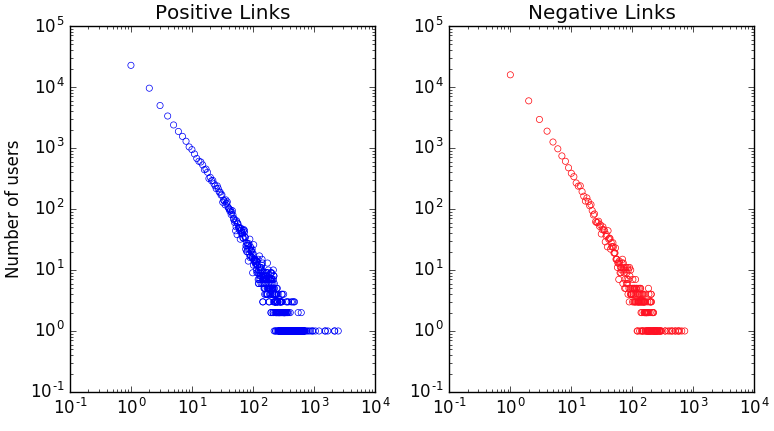}}			
	\end{center}
\vspace{-0.1in}
\caption{Degree Distributions in Signed Social Networks.}
\vspace{-0.1in}
\label{fig:degree-distributions}
\end{figure}

\subsection{Reciprocal Links in Signed Social Networks}

Links in directed social networks can be generally categorized into reciprocal (two-way) and parasocial (one-way) links~\cite{scott2012social}. Reciprocal links among nodes in unsigned networks are usually treated as the basis to create stable social ties and play an important role in the formation and evolution of networks~\cite{li2017understanding}. In this subsection, we study reciprocal links in signed social networks.  

For a pair of users $(u_i, u_j)$, there are four types of reciprocal links -- $(u_i + u_j, u_j + u_i)$, $(u_i + u_j, u_j - u_i)$, $ (u_i - u_j, u_j - u_i)$ and $(u_i - u_j, u_j + u_i)$, where $u_i + u_j$ (or $u_i - u_j$) denotes that there is a positive link (or a negative link) from $u_i$ to $u_j$. We checked our four signed networks and found that among four types of reciprocal links, there are few $(u_i + u_j, u_j - u_i)$ and $(u_i - u_j, u_j + u_i)$. Therefore, our analysis on reciprocal links focuses on $(u_i + u_j, u_j + u_i)$ and $ (u_i - u_j, u_j - u_i)$.  We calculate if $u_i$ has a positive link (or a negative link) to $u_j$, how likely $u_j$ also has a positive link (or a negative link) to $u_i$. The results on four signed networks are shown in Table~\ref{tab:reciprocal-links}. 

\begin{table}
\begin{center}
\caption{Reciprocal Links in Signed Social Networks.}
\label{tab:reciprocal-links}
\begin{tabular}{c|c|c}
	\hline
	 Datasets         &Positive Links  &Negative Links  \\	\hline
	 Bitcoin-Alpha  & 85.4\% & 18.0\% \\	
	 Bitcoin-OTC   & 83.8\%  & 17.8\% \\	
	 Slashdot   & 30.7\%  & 7.4\% \\	
	 Epinions  &34.8\% &3.8\%  \\	\hline
\end{tabular}

\end{center}
\end{table}

From the table, we make the following observations: 
\begin{itemize}
\item The percent of reciprocal positive links is much higher than that of reciprocal negative links in all four signed social networks; 
\item Though in all four websites, positive links are always visible to the public, the percent of reciprocal positive links in Bitcoin-Alpha and Bitcoin-OTC is much higher than that in Slashdot and Epinions. Users in Bitcoin Alpha and OTC exchange bitcoins with others; while users share free content (news or reviews) with others in Slashdot and Epinions. Thus, Bitcoin Alpha and OTC users need much stronger social ties for bitcoin trading in the online worlds than users in Slashdot and Epinions to consume online free content; and 
\item The percent of reciprocal negative links in Bitcoin-Alpha and Bitcoin-OTC is much higher than that in Slashdot, where the percent of reciprocal negative links in Slashdot is much higher than that in Epinions. Four websites have different access controls to negative links. In Bitcoin Alpha and OTC, negative links are totally visible to the public; only users who login to the Slashdot can see negative links; while negative links are totally private in Epinions. Exposing negative links may cause revenges that consequently could lead to more reciprocal negative links~\cite{szell2010multirelational}.  
\end{itemize}

\subsection{Balance Theory in Signed Networks}

Social theories such as homophily~\cite{mcpherson2001birds} play an important role in building node relevance measurements for unsigned social networks~\cite{liben2007link}. In this subsection, we investigate one of the most fundamental social theories related to signed social networks, i.e., balance theory~\cite{cartwright1956structural}, that could be helpful in building node relevance measurements in signed social networks.  

Generally, balance theory is based on the intuition that ``the friend of my friend is my friend" and ``the enemy of my enemy is my friend"~\cite{cartwright1956structural}.  We adopt $s_{ij}$ to denote the link sign between two users $u_i$ and $u_j$ where $s_{ij} = 1$ (or $s_{ij} = -1$) if the positive (or negative) link between $u_i$ and $u_j$. Balance theory suggests that a circle is balanced if there are even number of negative links. We typically focus on triads (or 3-circles)~\cite{leskovec2010predicting}. A triad of three users $(u_i, u_j, u_k)$ is balanced if $s_{ij} = 1$ and $s_{jk} = 1$, then $s_{ik} = 1$; or $s_{ij} = -1$ and $s_{jk} = -1$, then $s_{ik} = 1$. Therefore, for a triad, there are four possible sign combinations $(+,+,+)$, $(+,+,-)$, $(+,-,-)$ and $(-,-,-)$, while only  $(+,+,+)$ and $(+,-,-)$ are balanced.  Note that balance theory is only applicable to undirected signed network, we ignore the link directions when applying it to directed signed networks following the discussions in~\cite{leskovec2010predicting}. We count each of the four sign combinations and find that $92.0\%$, $91.5\%$, $94.5\%$ and $92.4\%$ of triads in Bitcoin-Alpha, Bitcoin-OTC, Slashdot and Epinions are balanced, respectively. 

\subsection{Discussions}

We summarize the observations from the above preliminary data analysis as below: 
\begin{itemize}
\item Properties of negative links could be different from positive links, which makes signed social networks be distinct from unsigned social networks. Therefore, though node relevance measurements have been extensively studied, it still needs dedicated efforts to systematically investigate signed relevance measurements. 
\item Most of triads in signed social networks satisfy balance theory. Thus, it can guide us to build advanced and novel signed relevance measurements. 
\end{itemize}

\section{Signed Node Relevance Measurements }
 
Node relevance measurements have been extensively studied in unsigned networks.  According to our preliminary data analysis in the last section, the availability of negative links makes signed networks unique in many aspects such as properties and balance theory. In this section, analogous to unsigned networks, we develop node relevance measurements for signed networks. 

\subsection{Notations and Definitions}

A signed network $\mathcal{G}$ is composed of a set of $N$ nodes (i.e., users) $\mathcal{U} = \{u_1, u_2, \dots, u_N\}$, a set of positive links $\mathcal{E}^+$ and a set of negative links $\mathcal{E}^-$. We represent signed links between users in an adjacency matrix, $\mathbf{A} \in \mathbb{R}^{N \times N}$, where $\mathbf{A}_{ij} = 1$ if $u_i$ has a positive link to $u_j$, $-1$ if $u_i$ creates a negative link to $u_j$, and $0$ when $u_i$ has no link to $u_j$. Furthermore, we can separate a signed network into two networks, one containing only positive links and the other with only negative links, which we can represent in the adjacency matrices $\mathbf{A}^+ \in \mathbb{R}^{N \times N}$ and $\mathbf{A}^- \in \mathbb{R}^{N \times N}$, respectively. We represent a positive link from $u_i$ to $u_j$ with $\mathbf{A}^+_{ij} = 1$ and  $\mathbf{A}^+_{ij} = 0$ otherwise. Similarly, we represent a negative link from $u_i$ to $u_j$ with $\mathbf{A}^-_{ij} = 1$ and  $\mathbf{A}^-_{ij} = 0$ otherwise.

We use $\mathbf{R} \in \mathbb{R}^{N \times N}$ to denote the relevance score matrix, where $\mathbf{R}_{ij}$ represents the node relevance from user $u_i$ to user $u_j$. Note that node relevance values are not necessarily symmetrical. We summarize the above notations in Table~\ref{tab:notations} where $d_i$ and $N_i$ denote degree and the set of neighbors of $u_i$ in an unsigned network. 

\begin{table}
\begin{center}
\caption{Notations.}
\label{tab:notations}
\resizebox{\columnwidth}{!}{
\begin{tabular}{|l|l|}
	\hline
	\label{tab:notations} 	 
	Notations & Descriptions\\
	\hline
	$\mathbf{R}$ & Node relevance matrix\\
	$\mathbf{A}$ & Adjacency matrix\\
	$\mathbf{A}^+ (\mathbf{A}^-)$ & Adjacency matrix of only positive(negative) links\\
	$|\mathbf{A}|$ & Absolute adjacency matrix \\
	$d_i$ & Degree of node $u_i$ \\
	$d^{in}_i (d^{out}_i$) & Indegree (Outdegree) of node $u_i$ \\
	$d^{in+}_i (d^{out+}_i)$ & Indegree (Outdegree) of positive links of node $u_i$ \\
	$d^{in-}_i (d^{out-}_i)$ & Indegree (Outdegree) of negative links of node $u_i$ \\
	$N_i$ & Set of neighbors for node $u_i$\\
	$N^{in}_i (N^{out}_i)$ & Set of incoming (outgoing) neighbors for node $u_i$ \\
	$N^{+}_i (N^{-}_i)$ & Set of positive (negative) neighbors for node $u_i$ \\
	$\mathbf{X}_{ij}$ & the (i,j) entry of the matrix $\mathbf{X}$\\
	\hline
\end{tabular}}
\end{center}
\end{table}

Many node relevance measurements have been proposed for unsigned networks. According to the used information, we can roughly divide them to local and global measurements. Local measurements only use local neighborhood information such as common neighbors; while global measurements utilize the whole structural information such as Random Walk with Restart.  Meanwhile, node relevance measurements can be undirected and directed, corresponding to undirected and directed networks. Note that we could use any method that requires a directed network for an undirected network, since undirected networks are simply directed networks where each edge has both directions. In this work, we will group signed relevance measurements as local and global methods. 

With node relevance measurements for unsigned networks, there are three strategies to design signed ones. The first is to only use $\mathbf{A}^+$ in the calculation of node relevance scores. This strategy completely ignores the negative links that could result in over-estimation of the impact of positive links~\cite{tang2016survey}.  The second strategy would be to convert negative links in the signed network into positive links, thus making the signed network into an unsigned network. Such a network can be represented by the matrix $\tilde{\mathbf{A}}$ where $\tilde{\mathbf{A}}_{ij} = |\mathbf{A}_{ij}|$. Ignoring signs of links not only overlooks the differences between negative and positive links; but also makes balance theory for signed networks not applicable. Our third strategy is to take advantage of negative links and balance theory to develop signed relevance measurements based on unsigned ones. In the following subsections, we will detail how to apply the third strategy to representative unsigned node relevance measurements. 

\subsection{Local Methods}

In this subsection, we build local signed relevance measurements based on representative local methods for unsigned networks including common neighbors, Jaccard Index, and Preferential Attachment~\cite{lorrain1971structural,newman2010intro}. For each unsigned measurement, we will first briefly introduce it, then detail how to design the signed one and finally discuss its connection with signed network properties and balance theory. 

\subsubsection{Common neighbors}

\textbf{Unsigned Common neighbors (UCN)}: If two nodes share a lot of common friends, they are likely to be relevant. Based on this intuition, UCN defines the relevance score between $u_i$ and $u_j$ as the number of common neighbors, which is formally defined as:  
\begin{align}
	\mathbf{R}_{ij} = |N_i \cap N_j|
\end{align}
\noindent where $|x|$ denotes the size of the set $x$. 

\textbf{Signed Common neighbors (SCN)}: UCN cannot be directly extended to include negative links. Therefore, we define SCN as follows:
	\begin{align}
	\mathbf{R}_{ij} = (|N^+_i \cap N^+_j| + |N^-_i \cap N^-_j|) \\- (|N^+_i \cap N^-_j| + |N^-_i \cap N^+_j|) \nonumber
	\end{align}
	We can interpret SCN  as number of common neighbors of $u_i$ and $u_j$ where they agree on the polarity of the sign ($|N^+_i \cap N^+_j| + |N^-_i \cap N^-_j|$) and then subtracting the number of neighbors that they disagree on the sign ($|N^+_i \cap N^-_j| + |N^-_i \cap N^+_j|$). 
	
	{\it Connection to Balance Theory:} If $u_i$ and $u_j$ agree with the majority of the signs of their neighbors, i.e.,  $ (|N^+_i \cap N^+_j| + |N^-_i \cap N^-_j|) >  (|N^+_i \cap N^-_j| + |N^-_i \cap N^+_j|)$, then $\mathbf{R}_{ij}$ is positive which will lead to more balanced triads. Otherwise, they have more disagreements on the signs, i.e., $ (|N^+_i \cap N^-_j| + |N^-_i \cap N^+_j|) > (|N^+_i \cap N^+_j| + |N^-_i \cap N^-_j|)$, then $\mathbf{R}_{ij}$ is negative, which will also result in more balanced triads. Therefore, SCN aims to force more triads with $u_i$ and $u_j$ to be balanced. 
	
\subsubsection{Jaccard Index} 

\textbf{Unsigned Jaccard Index (UJI)}: UCN only considers the number of common neighbors of $u_i$ and $u_j$, but it ignores the number of unique neighbors these two users have. Therefore, UCN is likely to give users with large numbers of neighbors high relevance scores. To mitigate such effect, UJI penalizes the UCN scores by the number of unique neighbors two users have as: 
	\begin{align}
		\mathbf{R}_{ij} = \frac{|N_i \cap N_j|}{|N_i \cup N_j|} 			
	\end{align}			

\textbf{Signed Jaccard Index (SJI)}:  Similar to from UCN to UJI, SJI is defined as SCN divided by the total number of unique neighbors $u_i$ and $u_j$ have: 
\begin{align}
\mathbf{R}_{ij} = \frac{SCN_{ij}}{|N_i^+ \cup N_i^- \cup N_j^+ \cup N_j^-|}
\end{align}

{\it Connection to Balance Theory}: Similar to SCN, SJI targets to force more triads balanced. 

\subsubsection{Preferential Attachment}

\textbf{ Unsigned Preferential Attachment (UPA)}: One commonly used interpretation behind this method, taken from the finance realm, is that the rich gets richer. In terms of social network analysis, users that already have many friends are more likely to create new friends in the future. Therefore, the node relevance score of UPA is to multiply the degrees of the two users~\cite{barabasi1999pa}.
\begin{align}
	\mathbf{R}_{ij} = d_i \times d_j
\end{align}
		
\textbf{Signed Preferential Attachment (SPA)}:  In the Section~\ref{data_analysis}, we demonstrate that both positive and negative links follow the power-law distributions. In other words, we observe ``the rich getting richer'' for both positive and negative links, which paves us a way to define SPA.  We first split the network from $\mathbf{A}$ to a positive network $\mathbf{A}+$ and a negative network $\mathbf{A}^-$. Then we can use UPA to calculate relevance scores from the positive and negative networks, separately, since degrees in both networks follow power-law distributions.  The relevance score for $i$ and $j$ from $\mathbf{A}^+$ is denoted as $UPA_{ij}^+$ and similarly we denote the relevance as $UPA_{ij}^-$ from $\mathbf{A}^-$.  $UPA_{ij}^+$ and $UPA_{ij}^-$ are computed as: 
	\begin{align}
		UPA^+_{ij}=d_i^{+} \times d_j^{+},~~~~UPA^-_{ij}=d_i^{-} \times d_j^{-} \nonumber
	\end{align}

Then we define SPA between $u_i$ and $u_j$ as: 
	\begin{align}
	    \mathbf{R}_{ij} = sign ( UPA^{+}_{ij} - UPA^{-}_{ij} ) f ( UPA^{+}_{ij} , UPA^{-}_{ij})
	\end{align}		 
\noindent where $sign(x) = $ 1, 0, or -1 if $x$ is larger, equal or smaller than $0$. Intuitively, if the positive relevance score $UPA^{+}_{ij}$ is larger than the negative one $UPA^{-}_{ij}$, the overall $\mathbf{R}_{ij}$ should be positive; otherwise,  $\mathbf{R}_{ij}$ should be negative. Therefore the sign of $\mathbf{R}_{ij}$ is decided by $sign ( UPA^{+}_{ij} - UPA^{-}_{ij} )$. The relevance strength $|\mathbf{R}_{ij} |$ is to aggregate $UPA^{+}_{ij}$ and $UPA^{-}_{ij}$ via a function $f$.  A straightforward way is to set $ f ( UPA^{+}_{ij} , UPA^{-}_{ij}) = | UPA^{+}_{ij} - UPA^{-}_{ij}|$.  It may not work well. For example, when $u_i$ and $u_j$ have both larger positive and negative degrees, positive and negative relevance scores will cancel each other, which contradicts with `` the rich getting richer".  Actually we empirically find that $ f ( UPA^{+}_{ij} , UPA^{-}_{ij}) = \max(UPA^{+}_{ij}, UPA^{-}_{ij})$ works better than $ f ( UPA^{+}_{ij} , UPA^{-}_{ij}) = | UPA^{+}_{ij} - UPA^{-}_{ij}|$. 

{\it Connection to the signed network property}:  According to the power-law distributions of positive and negative links, we design SPA, which will allow users with higher degrees to have higher relevance scores with others.

\subsection{Global Methods}
The global methods make use of not only the local neighborhoods, but also allow for the propagation of relevance information to pass through the whole network. Most of the global methods for unsigned networks assume that two users $u_i$ and $u_j$ should have high relevance if they have neighbors with high relevance. In this subsection, we detail how to design global signed relevance measurements based on representative unsigned ones and then connect them to balance theory.

\subsubsection{Katz}
     \textbf{Unsigned Katz (UK) }: This method sums over the collection of all paths from $i$ to $j$ and has an exponential decay on the weight associated with the count of paths as the length increases~\cite{katz1953sim}: 
	\begin{align}
		\label{eq:unsigned_katz}
		\mathbf{R}_{ij}& = \sum\limits_{l=1}^{\infty} \beta^{l} \cdot |\text{paths}^{l}_{i,j}| = \sum\limits_{l=1}^{\infty} \beta^{l} \mathbf{A}^l
	\end{align}
	 where $|\text{paths}^{l}_{i,j}|$ is the count of paths of length $l$ from $i$ to $j$. Note that we should have $\beta < 1$ so that longer paths will be assigned less weight than shorter paths. This can be formulated recursively as follows to handle the counting of the paths of varying length:
	\begin{align}
		\mathbf{R}_{ij} = \frac{\beta}{d_x} \sum\limits_{k=1}^{N}\mathbf{A}_{ik}  \mathbf{R}_{kj} + \delta_{ij} 	
	\end{align}
	Note that $\delta_{ij}$ is used to ensure that every node in the network has a high relevance to themselves (i.e., ``self-similarity''). It is a diagonal term and is defined as $\delta = \mathbf{I}$. It normalizes the relevance scores from each user $u_i$ based on the degree $d_i$.

	\textbf{Signed Katz (SK):} Balance theory states that a k-cycle in a signed social network is balanced if it contains an even number of negative edges and unbalanced if it contains an odd number of negative edges. With relevance scores from SK, we expect more balanced k-cycles than unbalanced ones involving users $i$ and $j$. To achieve this, we would therefore need to choose the sign of the node relevance $\mathbf{R}_{ij}$ to be either positive or negative, such that it optimizes over all the cycles involving $i$ and $j$ (i.e., all the paths between $i$ and $j$). As done in UK, we also can similarly allow the decay of importance on the longer paths. Our formulation is shown below with its recurrence relation for the calculation of paths of length $l$ having an even or odd number of negative edges.
	
	\begin{align}
		\label{eq:signed_katz}
		\mathbf{R} = \sum\limits_{l=1}^{\gamma} \beta^l f(\mathbf{B}_l,\mathbf{U}_l)
	\end{align}	
	where
	\begin{align}
		\mathbf{B}_l = \mathbf{B}_{l-1}\mathbf{A}^+ + \mathbf{U}_{l-1}\mathbf{A}^- \nonumber \\ \nonumber
		\mathbf{U}_l = \mathbf{B}_{l-1}\mathbf{A}^- + \mathbf{U}_{l-1}\mathbf{A}^+ \\ \nonumber
		\mathbf{B_1} = \mathbf{A^+}, ~~~~ \mathbf{U_1} = \mathbf{A^-} \nonumber \nonumber
	\end{align}
\noindent where $f(\mathbf{B}_l,\mathbf{U}_l)$ is a function to combine the counts of paths with even and odd number of negative links. $\mathbf{B}_l$ and $\mathbf{U}_l$ are the matrices that hold the number of paths with an even and odd number of negative links in paths of length $l$, respectively.  Next we will discuss the inner working of SK.  When counting paths of length 1 (i.e., a direct edge connecting the two nodes),  we set $\mathbf{B}_1$ as $\mathbf{A}^+$ since having a positive edge is trivially having an even number of negative links in a path of length 1, and similarly reasoned for initializing $\mathbf{A}^-$.  We assume that $\mathbf{B}_{l-1}$ and $\mathbf{U}_{l-1}$ represent the paths of length $l-1$ having an even and odd number of negative edges, respectively, between all pairs of nodes.  Adding one positive link ($\mathbf{A}^+$) to a path in $\mathbf{B}_{l-1}$ or adding a negative link ($\mathbf{A}^-$ ) to a path in $\mathbf{U}_{l-1}$ will result in a path of length $l$ with an even number of negative links. This intuition leads to the update rule of  $\mathbf{B}_l = \mathbf{B}_{l-1}\mathbf{A}^+ + \mathbf{U}_{l-1}\mathbf{A}^-$. Similarly, we can obtain the update rule of $\mathbf{U}_l = \mathbf{B}_{l-1}\mathbf{A}^- + \mathbf{U}_{l-1}\mathbf{A}^+$. 

\begin{theorem}
\label{th:signed_katz}
When we choose $f(\mathbf{B}_l,\mathbf{U}_l) = (\mathbf{B}_l - \mathbf{U}_l)$ and $\mathbf{A} \in \mathbb{R}^{N \times N}$, where $\mathbf{A}_{ij} = 1$ if $u_i$ has a positive link to $u_j$, $-1$ if $u_i$ creates a negative link to $u_j$, and $0$ when $u_i$ has no link to $u_j$, signed Katz in Eq~(\ref{eq:signed_katz}) is equivalent to applying unsigned Katz in Eq~(\ref{eq:unsigned_katz}) on the signed network adjacency matrix defined as $\mathbf{A}$.
\end{theorem}
\begin{proof}
	To prove the theorem, we only need to show that: $ \mathbf{B}_l - \mathbf{U}_l = {\bf A}^l $. We use mathematical induction as: \\
	Basis: Let $l=1$,  based on our definition of $\mathbf{B}_1$ and $\mathbf{U}_1$, we have $(\mathbf{B}_1 - \mathbf{U}_1) = (\mathbf{A}^+ - \mathbf{A}^-) = \mathbf{A} = \mathbf{A}^l$.\\
	Inductive Hypothesis: Suppose the theorem holds for $l=k$. In other words,  $(\mathbf{B}_k - \mathbf{U}_k) = {\bf A}^k$. \\
	Inductive Step: Let $l=k+1$. Then our left size is $(\mathbf{B}_{k+1} - \mathbf{U}_{k+1}) = \Big( (\mathbf{B}_{k}\mathbf{A}^+ + \mathbf{U}_{k}\mathbf{A}^-) - (\mathbf{B}_{k}\mathbf{A}^- + \mathbf{U}_{k}\mathbf{A}^+) \Big) = (\mathbf{B}_k - \mathbf{U}_k)(\mathbf{A}^+ - \mathbf{A}^-) = \mathbf{A}^k(\mathbf{A}) = \mathbf{A}^{k+1}$, which completes the proof. 
\end{proof}

{\it Connection to Balance Theory}: SK is built based on balance theory. SCN and SJI forces more balanced triads (or 3-cycles), while SK pushes more for any $l$-circles to be balanced. If the majority of paths between $i$ and $j$ have an even number of negative links, according to balance theory, we should have a positive node relevance between them. Similarly, when having an odd number of negative edges, we want to have a negative relevance. Therefore, if we count the number of paths between $i$ and $j$ with an even or odd number of negative edges, then we can subtract the number with an odd number of negative links from the number of paths having an even number of links, since this will give us the optimal choice of sign between $i$ and $j$ as mentioned above. More specifically, if the resulting value is positive, the node relevance between $i$ and $j$ is positive, otherwise negative.

\subsubsection{Asymmetric Similarity Measure for Weighted Networks}

\textbf{ Unsigned Asymmetric Similarity Measure for Weighted Networks (UASCOS++)}: This method is an enrichment of the ASCOS \cite{chen2013ascos} to handle weighted networks. The formulation of ASCOS is the following:
\begin{equation}
  \mathbf{R}_{ij} = \begin{cases}
	\frac{c}{|N^{in}_i|} \sum\limits_{k \in N^{in}_i} \mathbf{R}_{kj} & i \neq j \\
	1 & i = j
  \end{cases} \nonumber
\end{equation}

Let $\mathbf{P}_{ij} = \frac{\mathbf{A}_{ij}}{d^{in}_i} $  and  we can rewrite the formulation as:  
\begin{align}
	\mathbf{R} = c\mathbf{P}^{\top}\mathbf{R} + (1-c)\mathbf{I} \nonumber
\end{align}

It defines the node relevance as the summation of normalized relevance from the incoming neighbors of $i$ to $j$. The modifications for ASCOS++ were performed to handle weights on the edges. The formulation is shown below:
\begin{equation}
  \mathbf{R}_{ij} = \begin{cases}
	c\sum\limits_{k \in N^{in}_i} \frac{\mathbf{A}_{ik}}{\sum\limits_{q \in N^{in}_i} \mathbf{A}_{iq}} (1-e^{-A_{ik}})  \mathbf{R}_{kj} & i \neq j \\
	1 & i = j
  \end{cases} 
\end{equation}
	The adjustment is that they now normalize each of the edge weights coming into $i$ by the summation of all the incoming weights into $i$. The term $(1-e^{-A_{ik}})$ maps the weights to be close to 1 when edge weights are large, and when the weights are small,  it maps them close to 0. 
	
\textbf{Signed ASCOS++ (SASCOS++)}: ASCOS++ has difficulties to directly adapt to signed networks.  Assume that a node $i$ has an even number of incoming edges, where half the edges are positive, while the other half are negative. Therefore, this would lead to an undefined value as the summation over all incoming edges to $i$ $\sum\limits_{q \in N^{in}_i} \mathbf{A}_{iq}$ is zero.	
	
Another issue is if we directly apply ASCOS++, the resulting relevance score could contradict with balance theory.   To ease our analysis in the following case, let $\kappa = \sum\limits_{q \in N^{in}_i} \mathbf{A}_{iq}$, $\lambda = \frac{\mathbf{A}_{ik}}{\kappa}$ and $\mu = (1-e^{-A_{ik}})$.  If $\mathbf{A}_{ik}=1$ and $\kappa$ is negative, hence $\lambda$ is negative and $\mu$ is positive. Thus, if $\mathbf{R}_{kj}$ is also positive, then the product of these three terms  $ \mathbf{R}_{ij}$ is negative and the resulting triad ($+$, $+$, $-$) does not follow balance theory. Similarity, when $\mathbf{R}_{kj}$ is negative, the product is positive and the resulting triad ($+$,$-$,$+$) is also not balanced. 

Due to the fact using ASCOS++ with signed networks, could inherently disagree with balance theory, which motivates us to build SASCOS++.  We note that when using ASCOS++ with signed networks, $\mu$ is equal to approximately 0.63 and -1.72 when $\mathbf{A}_{ik}$ is positive or negative, respectively. Thus, it is providing a stronger push in the similarity (by about three times) when seeing a negative link. Due to the imbalance of the numbers of positive and negative links in signed networks, we leave this $\mu$ term as is, but make a change to the normalization (i.e., $\kappa$). The formulation for SASCOS++ is shown below:
\begin{equation}
  \mathbf{R}_{ij} = \begin{cases}
	c\sum\limits_{k \in N^{in}_i} \frac{\mathbf{A}_{ik}}{\sum\limits_{q \in N^{in}_i} |\mathbf{A}_{iq}|} (1-e^{-A_{ik}})  \mathbf{R}_{kj} & i \neq j \\
	1 & i = j
  \end{cases} 
\end{equation}
{\it Connection to Balance Theory}: It is easy to verify that SASCOS++ is able to have the relevance measurements aligning with balance theory.	In other words, it will push more balanced triads. 
	
\subsubsection{Random Walk with Restart}
	\textbf{ Unsigned Random Walk with Restart (URWR)}:  A random walker starting on node $i$ that has a probability of $(1-c)$ to return to $i$ and with probability $c$ chooses a neighbor of the current node to move to based on a transition matrix $\mathbf{W}$ (where $\mathbf{W}_{ij}$ is the probability that the walker starting at $i$ will end at node $j$). We define this transition matrix as $\mathbf{W}_{ij} = \frac{1}{d_{i}}$ if $i$ and $j$ are connected and $\mathbf{W}_{ij} = 0$ otherwise (i.e., no link between $i$ and $j$). With the intuition, URWR is formulated as~\cite{tong2006fast}: 
	\begin{align}
		\mathbf{R} = c\mathbf{W} \mathbf{R} + (1 - c) {\bf I}
		= (1 - c)(\mathbf{I} - c\mathbf{W}^\top)^{-1}
	\end{align}

\textbf{Signed Random Walk with Restart (SRWR)}: The transition matrix $\mathbf{W}$ has to be non-negative, thus we cannot directly apply URWR to signed networks.  Therefore, we study signed random walk with restart.  Based on balance theory, the relevance score of $u_k$ w.r.t $u_i$ can be useful to infer that of $u_j$ to $u_i$ if there's a link from $u_k$ to $u_j$. For example, if $\mathbf{A}_{kj}>0$ (or $u_k$ and $u_j$ are friends), and $\mathbf{R}_{ik}>0$ (or $u_i$ and $u_k$ are likely to be friends), it may suggest that $u_i$ and $u_j$ are friends (or $\mathbf{R}_{ij}>0$) because friends' friends are friends. On the contrary, if $\mathbf{A}_{kj}<0$ (or $u_k$ and $u_j$ are enemies) but $\mathbf{R}_{ik}>0$ (or $u_i$ and $u_k$ are likely to be friends), it may indicate that $u_i$ and $u_j$ are enemies (or $\mathbf{R}_{ij}<0$) because friends' enemies are enemies, which is implied from ``the enemy of my enemy is my friend". This indicates that (1) $u_j$'s relevance score to $u_i$ can be indicated by these of nodes (e.g., $u_k$) that have links to ${u}_j$; and (2) the estimation also depends on the signs of links from $u_k$ to $u_j$ and the relevance scores from $u_i$ to $u_k$. These intuitions suggested by balance theory pave us a way to build SRWR.  Let $\bar{\mathbf{D}}$ be a diagonal matrix with its diagonal element $\bar{\mathbf{D}}_{ii}$ given as
\vskip -1.2em
\begin{equation}
    \bar{\mathbf{D}}_{ii} = \sum_{k} |\mathbf{A}_{ik}| \nonumber
\end{equation}
\vskip -0.2em
Apparently, $\bar{\mathbf{D}}_{ii}$ is the out degree of $u_i$ considering both positive and negative links. Thus, the normalized weight of the link from $u_i$ to $u_k$ is given as
\vskip -1em
\begin{equation}
    \bar{\mathbf{W}}_{ik} = \frac{|\mathbf{A}_{ik}|}{\bar{\mathbf{D}}_{ii}} \nonumber
\end{equation}

According to aforementioned intuitions, $\mathbf{R}_{ik}$ can be used to estimate $\mathbf{R}_{ij}$ with $\mathbf{A}_{kj} \ne 0$. Intuitively the portion of relevance score of $u_k$ contributes to $\mathbf{R}_{ij}$ should be weighted by $\bar{\mathbf{W}}_{ij}$. This is to account for the number of neighbors of $u_k$. If $\bar{\mathbf{D}}_{ii}$ is large, then $\bar{\mathbf{W}}_{ij}$ is small and the effects of $u_i$ to each of its neighbor is small. Thus, $\mathbf{R}_{ij}$ can be estimated as:
\begin{equation}
\label{eq:Rij_orig}
	\mathbf{R}_{ij} \propto \sum_{k} sign(\mathbf{A}_{kj}) \bar{\mathbf{W}}_{kj} \mathbf{R}_{ik}
\end{equation}
where $sign(\mathbf{A}_{kj})$ is used to encode the impact of the sign of the links. With sign introduced in the estimation of $\mathbf{R}_{ij}$, the relevance score can be both positive and negative. Two users with negative links can affect each other with negative relevance scores and thus can capture the semantic meanings of signed links.

With the analysis above, we are ready to discuss the details of SRWR. We focus on the relevance score of $u_j, j=1,\dots,n, j\ne i$ w.r.t $u_i$ since the relevance scores w.r.t other nodes can be derived similarly. Firstly, $\mathbf{R}_{ij}, j=1,\dots,n, j\ne i$, are initialized to 0, which means that the relevance scores of $u_j$ to $u_i$ is unknown; while $\mathbf{R}_{ii}$ is initialized to $1$ because $u_i$ should be positively relevant to itself. Now considering that a random walker starting from $u_i$. It can iteratively transmit to its neighborhood through positive and negative outgoing links. Each time the walker arrives at a node $u_j$, it will update $\mathbf{R}_{ij}$ by the relevance scores of nodes that have links to $u_j$.  If the random walker arrives at $u_i$, then $\mathbf{R}_{ii}$ is updated as
\begin{equation}
\label{eq:Rii_orig}
	\mathbf{R}_{ii} \leftarrow c\sum_{k} sign(\mathbf{A}_{ki}) \bar{\mathbf{W}}_{ki} \mathbf{R}_{ik} + (1-c) * 1
\end{equation}
where the first term of the right-hand side of Eq.(\ref{eq:Rii_orig}) is the relevance score estimated from neighborhood, and the second term is to make sure that $\mathbf{R}_{ii}>0$, i.e., $u_i$ is relevant to itself. $c$ is a scalar between $0$ and $1$, which is used to control the contribution of the two parts. If the random walker arrives at $u_j, j \ne i$, $\mathbf{R}_{ij}$ is updated as
\begin{equation}
\label{eq:Rij_modi}
	\mathbf{R}_{ij} \leftarrow c\sum_{k} sign(\mathbf{A}_{kj}) \bar{\mathbf{W}}_{kj} \mathbf{R}_{ik}
\end{equation}
Combining Eq.(\ref{eq:Rij_orig}) and Eq.(\ref{eq:Rii_orig}) together, $\mathbf{R}_{ij}$ is updated as
\begin{equation}
	\mathbf{R}_{ij} \leftarrow c\sum_{k} sign(\mathbf{A}_{kj}) \bar{\mathbf{W}}_{kj} \mathbf{R}_{ik} + (1-c)\mathbb{I}(i,j) \nonumber
\end{equation}
where $\mathbb{I}(i,j)$ is a binary indicator function with $\mathbb{I}(i,j) = 1$ if $i=j$ and $0$ otherwise. The random walker keeps moving until $\mathbf{R}$ doesn't change, which gives
\begin{equation}
\label{eq:Rij}
	\mathbf{R}_{ij} = c\sum_{k} sign(\mathbf{A}_{kj}) \bar{\mathbf{W}}_{kj} \mathbf{R}_{ik} + (1-c)\mathbb{I}(i,j)
\end{equation}
By noticing that $sign(\mathbf{A}_{kj}) \bar{\mathbf{W}}_{kj} = \frac{\mathbf{A}_{kj}}{\bar{\mathbf{D}}_{kk}}$, we define $\mathbf{S}$ as
\begin{equation}
	\mathbf{S} = \bar{\mathbf{D}}^{-1} \mathbf{A}
\end{equation}
and then Eq.(\ref{eq:Rij}) can be written in matrix form as
\begin{equation}
    \mathbf{R} = c \mathbf{R} \mathbf{S} + (1-c) \mathbf{I}
\end{equation}
where $\mathbf{I}$ is the identity matrix. The solution to the above equation is given as
\begin{equation}
\label{eq:signed}
    \mathbf{R} = (1-c)(\mathbf{I} - c\mathbf{S})^{-1}
\end{equation}

\begin{figure}[!t]
\centering
	\subfigure[+++]{
		\label{fig:triplet1}
		\includegraphics[scale=0.2]{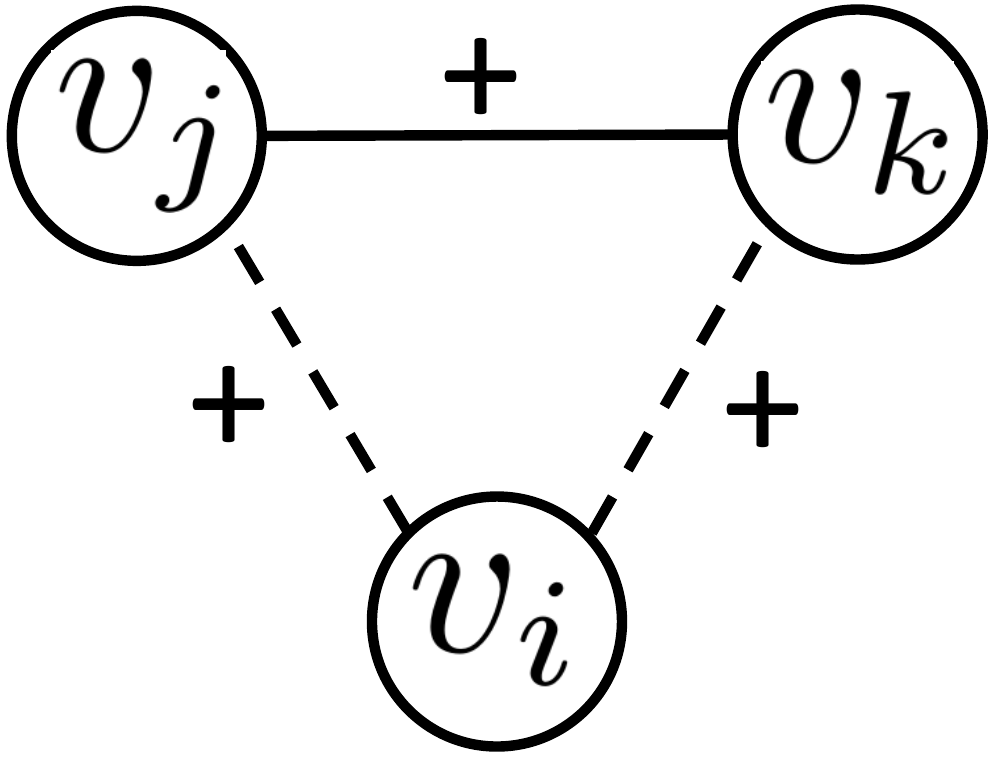}}
	\subfigure[+ - +]{
		\label{fig:triplet2}
		\includegraphics[scale=0.2]{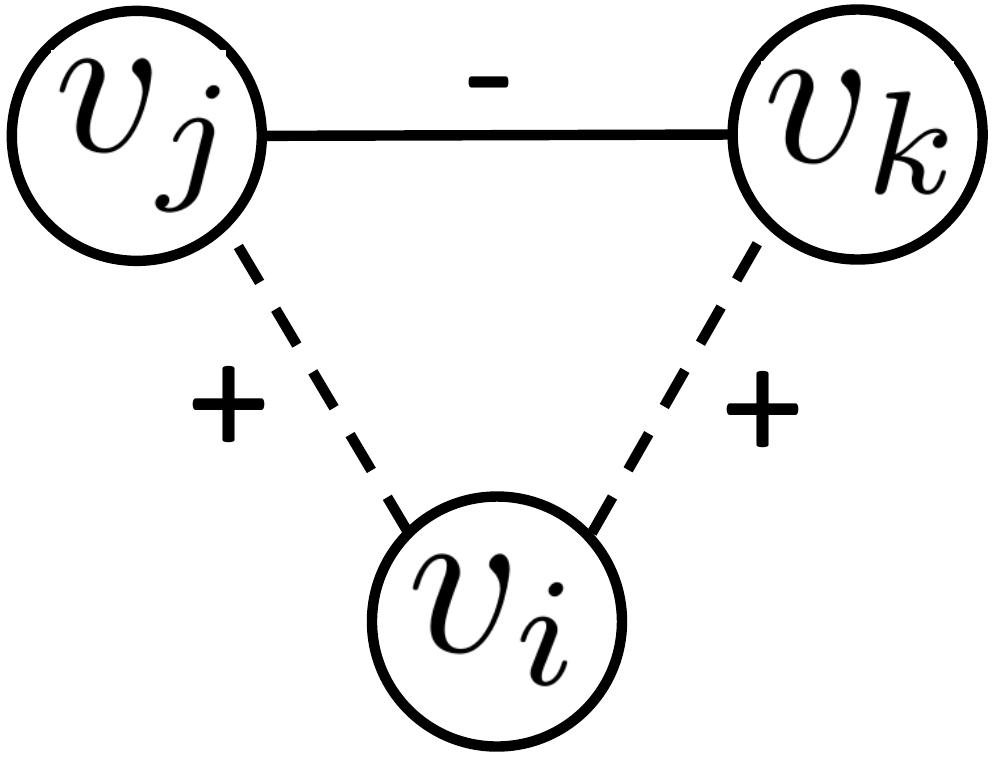}}
	\subfigure[++ -]{
		\label{fig:triplet3}
		\includegraphics[scale=0.2]{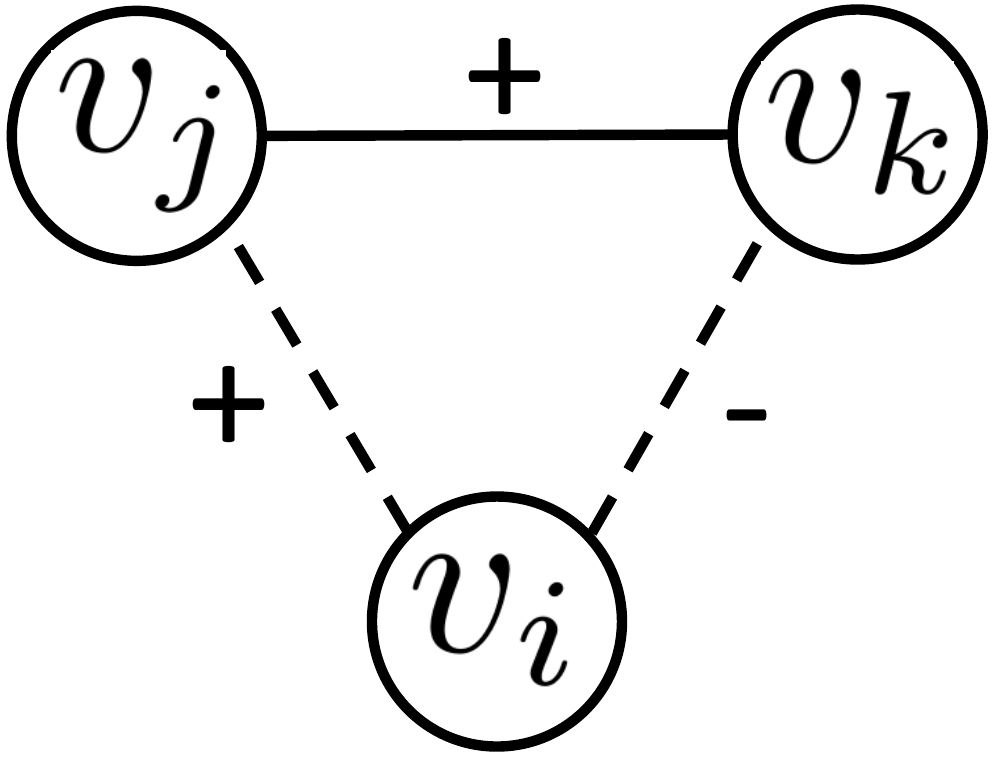}}
	\subfigure[+ - - ]{
		\label{fig:triplet4}
		\includegraphics[scale=0.2]{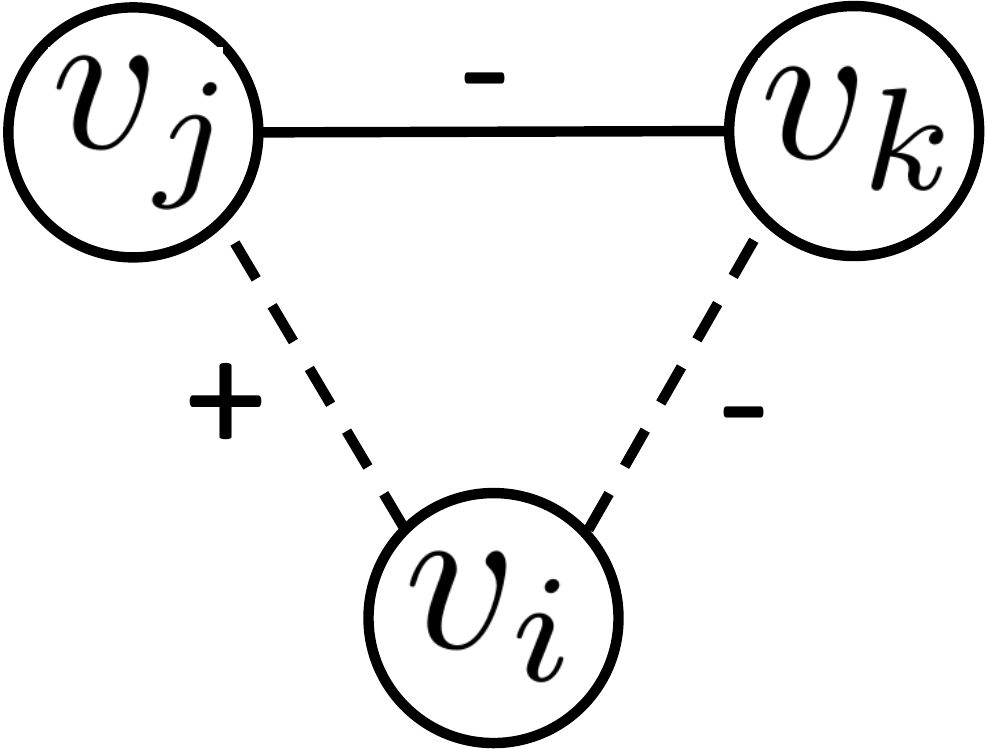}}
	\subfigure[- -+]{
		\label{fig:triplet6}
		\includegraphics[scale=0.2]{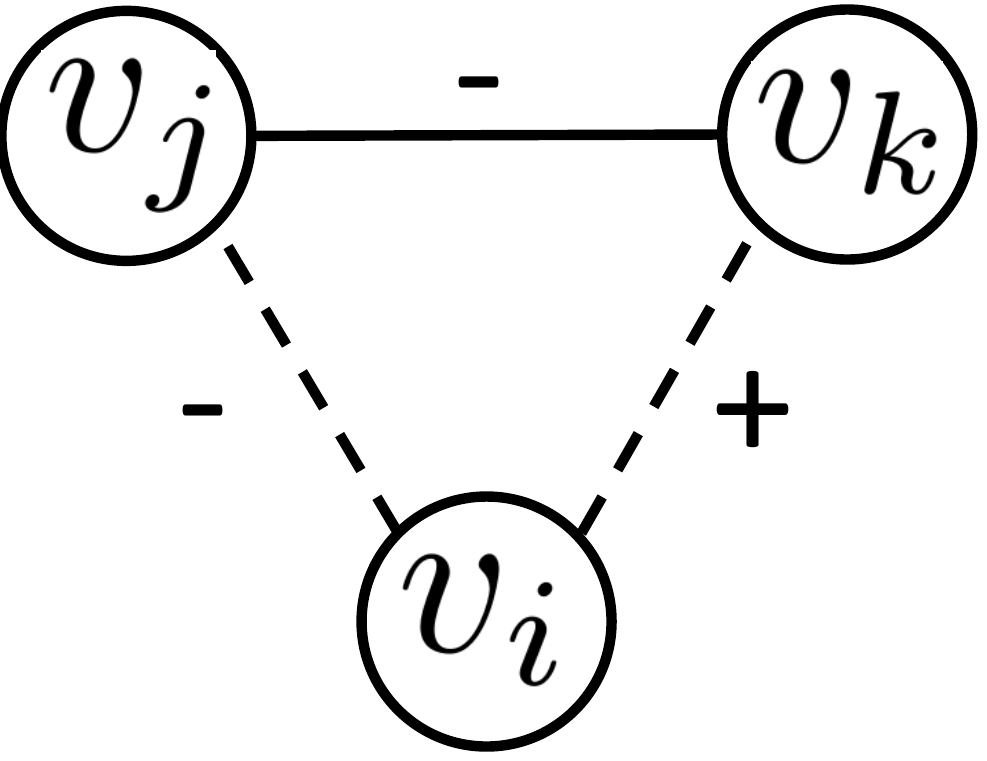}}
	\subfigure[- - -]{
		\label{fig:triplet5}
		\includegraphics[scale=0.2]{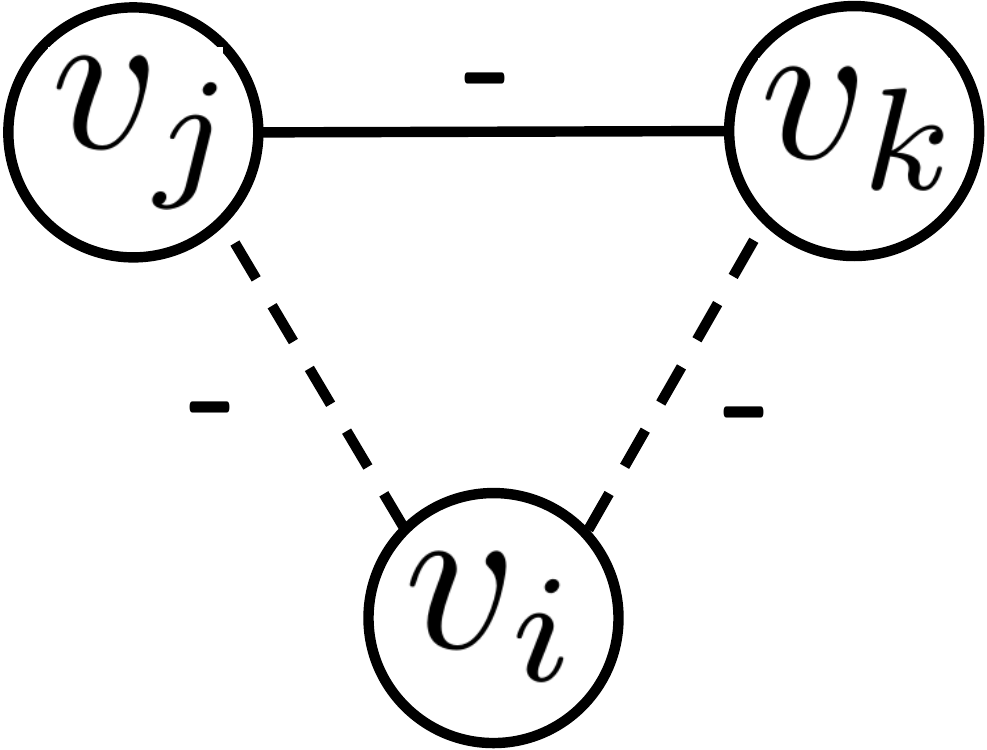}}
	\vskip -1em
	\caption{Triplets Encountered During Random Walk}
	\label{fig:triplets}
	\vskip -1em
\end{figure}

{\it Correctness}: Here we show that SRWR is correct, i.e., $(\mathbf{I} - c \mathbf{S})^{-1}$ exists. The existence of $(\mathbf{I} - c \mathbf{S})^{-1}$ can be proofed using the following lemma, which is known as Levy-Desplanques theorem~\cite{horn2012matrix}. The Levy-Desplanques theorem is stated as follows
\begin{lemma}
Let $\mathbf{P} \in \mathbb{R}^{n \times n}$ be a square matrix.If $|\mathbf{P}_{ii}| > \sum_{j \ne i}|\mathbf{P}_{ij}|$ for all $i=1,\dots,n$, then $\mathbf{P}$ is nonsingular.
\end{lemma}
\noindent{}Based on the above lemma, we have
\begin{theorem}
\label{th:nonsingular}
	$\mathbf{I} - c \mathbf{S}$, $ 0 < c < 1$, is non-singular.
\end{theorem}
\begin{proof}
	Let $\mathbf{P} = \mathbf{I} - c \mathbf{S}$. Since $\mathbf{S}_{ii} = 0$, we have $\mathbf{P}_{ii} = 1$. Also, $\sum_{j \ne i} |\mathbf{S}_{ij}| $ is given as
	\begin{equation}
		\sum_{j \ne i} |\mathbf{S}_{ij}| = \sum_{j} |\mathbf{S}_{ij}| = \sum_j \frac{|	\mathbf{A}_{ij}|}{\bar{\mathbf{D}}_{ii}} = \frac{\sum_j |\mathbf{A}_{ij}|}	{\bar{\mathbf{D}}_{ii}} = 1.
	\end{equation}
which leads to $\sum_{j \ne i}|\mathbf{P}_{ij}| = c \sum_{j \ne i} |\mathbf{S}_{ij}|=c$. Then we have $|\mathbf{P}_{ii}| > \sum_{j \ne i}|\mathbf{P}_{ij}|$ for all $i=1,\dots,n$. Thus, $\mathbf{I} - c \mathbf{S}$ is non-singular and  $(\mathbf{I} - c \mathbf{S})^{-1}$ exists.
\end{proof}

{\it Connection to balance theory}: Figure \ref{fig:triplets} gives representative triplets that will happen during the update process. The solid line with +/- means positive/negative links. The dashed line with +/- means $\mathbf{R}_{ij}>0/\mathbf{R}_{ij}<0$. According to the social balance theory~\cite{wasserman1994social}, the resulting triads in Figures \ref{fig:triplet1}, \ref{fig:triplet4} and \ref{fig:triplet6} are balanced while the remaining three are unbalanced. Next we show that SRWR is likely to keep the balanced structures while reducing unbalanced structures during the updating process. For example, in Figure \ref{fig:triplet1}, $\mathbf{R}_{ik}\mathbf{S}_{kj}>0$ will be added to $\mathbf{R}_{ij}$ according to Eq.~(\ref{eq:Rij_modi}), which increases the positive relevance score $\mathbf{R}_{ij}$.  However, in Figure \ref{fig:triplet2}, $\mathbf{R}_{ik}\mathbf{S}_{kj}<0$ will be added to $\mathbf{R}_{ij}$ that reduces the positive relevance score $\mathbf{R}_{ij}>0$. $\mathbf{R}_{ij}$ will be consistently reduced until $\mathbf{R}_{ij}$ becomes negative (or the triad becomes balanced). Following a similar process, we can give similar observations for other triads. Thus, SRWR actually tends to learn relevance scores that increase the structural balance of a given signed network.

\begin{table}
	\begin{center}
	\small
	\vskip -1em
	\caption{\label{tab:link_undirected}Performance comparison of link prediction under the undirected setting.}
		\begin{tabular}{|c|c|c|c|c|}	\hline
			Metrics		&	\begin{tabular}{@{}c@{}}Bitcoin- \\ Alpha\end{tabular}	&	\begin{tabular}{@{}c@{}}Bitcoin- \\ OTC\end{tabular}	&	Slashdot		&	Epinions \\ \hline		
			UCN-R		&	0.523			&	0.500		&	0.520		&	0.520 \\
			UCN-I		&	0.497			&	0.501		&	0.508		&	0.508 \\
			SCN			&	0.716			&	0.671		&	0.549		&	0.629 \\
			UJI-R		&	0.524			&	0.499		&	0.513		&	0.522 \\
			UJI-I		&	0.489			&	0.497		&	0.503		&	0.512 \\						
			SJI			&	0.725			&	0.669		&	0.550		&	0.630 \\
			UPA-R		&	0.587			&	0.497		&	0.571		&	0.634 \\
			UPA-I		&	0.475			&	0.481		&	0.484		&	0.498 \\
			SPA			&	0.628			&	0.559		&	0.641		&	0.634 \\ \hline
			UK-R			&	0.587			&	0.517		&	0.542		&	0.560 \\
			UK-I		&	0.482			&	0.488		&	0.498		&	0.538 \\
			SK			&	0.766	&	0.730		&	\textbf{0.693}	&	0.702 \\ 
			
			URWR-R		&	0.628			&	0.531		&	0.569		&	0.566 \\
			URWR-I		&	0.481			&	0.500		&	0.494			&	0.530 \\
			SRWR			&	\textbf{0.775}			&	0.751	&	0.677	&	0.703 \\
			UASCOS++-R	&	0.603			&	0.530		&	0.554		&	0.573 \\			
			UASCOS++-I	&	0.484			&	0.496		&	0.497		&	0.537 \\			
			SASCOS++		&	0.774			&	\textbf{0.765}	&	0.663		&	\textbf{0.705} \\ \hline
		\end{tabular}
	\end{center}

	\vskip 2em
	\begin{center}
	\small
	\vskip -1em
	\caption{\label{tab:link_directed}Performance comparison of link prediction under the directed setting.}
		\begin{tabular}{|c|c|c|c|c|}	\hline
			Metrics		&	\begin{tabular}{@{}c@{}}Bitcoin- \\ Alpha\end{tabular}	&	\begin{tabular}{@{}c@{}}Bitcoin- \\ OTC\end{tabular}	&	Slashdot		&	Epinions \\ \hline		
			UASCOS++-R	&	0.630			&	0.588		&	0.524		&	0.516 \\
			UASCOS++-I	&	0.639			&	0.562			&	0.519		&	0.493 \\			
			SASCOS++		&	0.705			&	0.644		&	0.578		&	0.580 \\			
			URWR-R		&	0.644			&	0.606		&	0.541		&	0.565 \\		
			URWR-I		&	0.590			&	0.556		&	0.500		&	0.563 \\
			SRWR			&	\textbf{0.809}	&	\textbf{0.791}	&	\textbf{0.627}	&	\textbf{0.687} \\	 \hline
		\end{tabular}
	\end{center}
\end{table}

\section{Experiment}

In this section, we investigate the impact of signed relevance measurements on two signed network analysis tasks, i.e., link prediction and tie strength prediction. We aim to answer the following two questions.  As mentioned in the last section, we can have three strategies to adapt unsigned measurements for signed networks -- (1) removing negative links; (2) ignoring signs; and (3) building advanced signed versions based on signed network properties and balance theory.  Note that in the following subsections, given an unsigned measurement ``X'', we use ``X-R" and ``X-I" to denote the corresponding measurements applicable to signed networks by removing negative links and ignoring signs, respectively.  For example, ``UCN-R" and ``UCN-I'' denote the strategies of adapting ``UCN'' to signed networks by removing negative links and ignoring signs, separately.  The first question we want to answer is -- which strategy leads to better measurements.  We have built numerous local and global measurements. The second question is -- how they perform in different tasks. 

For each of the parameterized measurements, we performed cross validation for the parameter tuning for each of the tasks. Among measurements discussed in the last section, common neighbor (CN), Jaccard Index (JI), and Preferential Attachment (PA)
based measurements are designed for undirected networks; while ASCOS and RWR are for directed networks. As mentioned before directed measurements can be naturally applied to undirected ones by considering one undirected link as two directed links. Therefore, we conduct experiments with both undirected and directed settings.  

\subsection{Link Prediction} \label{link_prediction}

The problem of link prediction in signed networks is to predict new positive and negative links by given old positive and negative links.  Previous study in unsigned networks suggested that good node relevance measurements generally are good for the prediction of links~\cite{lu2011link}. Therefore, the link prediction performance can reflect the quality of relevance measurements. 

\subsubsection{Experimental Settings}

For each dataset, we randomly choose $80\%$ as training, and the remaining as testing. We perform relevance measurements on the training set to get the relevance scores for each pair of users. The signed specific measurements can obtain a relevance score from $[-1,1]$; hence we directly use the sign of the relevance score to indicate the sign of links.  For ``X-R" and ``X-I", the relevance score is in ``[0,1]". From the training data, we search an optimal threshold from the training data, and then if the relevance score is less than threshold, we predict a negative link and positive otherwise. Since positive and negative links are usually imbalanced in real-world signed networks, we use Area Under the Curve (AUC) as the metric to assess the performance of link prediction. For all four datasets, network information is available thus they all can be used in the link prediction experiment.  Under the undirected setting, we ignore the directions of links following common practice in~\cite{leskovec2010predicting}.

\subsubsection{Link Prediction Performance}

The link prediction comparison results are shown in Table~\ref{tab:link_undirected} and Table~\ref{tab:link_directed} for undirected and directed settings, respectively. From the Table~\ref{tab:link_undirected}, we make the following observations under the undirected setting: 
\begin{itemize}
\item Signed specific relevance measurements perform much better than these by (1) removing negative links and (2) ignoring signs.  These results suggest the importance of negative links in building node relevance measurements for signed networks.   
\item Global signed measurements consistently obtain better link prediction performance than local signed measurements. Global methods consider long circles; while local methods only consider triads. This observation is consistent with that in~\cite{chiang2011exploiting} -- long circles contain rich information in helping predict the signs of links. 
\end{itemize}

Under the directed setting, signed measurements also outperform than these via (1) removing negative links and (2) ignoring signs; while the signed RWR obtains the best performance. 

\subsection{Tie Strength Prediction}

The relevance score for signed networks not only can indicate the signs of links but also can indicate the connection strengthen. Therefore, another possible application of relevance measurements is tie strength prediction, which aims to assign a weight to a link to indicate the connection strengthen~\cite{gilbert2009predicting,xiang2010modeling,kahanda2009strength}. In other words, the input of a tie strength prediction algorithm is an unweighted (or binary) network and the output is a weighted network. 

\subsubsection{Experimental Settings}
We have only used the two Bitcoin datasets (Bitcoin-Alpha and Bitcoin-OTC) for this task as they are the only two of the four datasets that have a ground truth strength associated with each edge in the network. Note that we have normalized the two datasets to have their strength in the range [-1,1] to ensure easy mappings from our presented node relevance measurements to the tie strengths associated with these datasets edges.

We directly use the relevance scores of signed specific measurements as the predicted tie strength. While for ``X-R'' and ``X-I", we use the similar strategy as link prediction for tie strength prediction -- we search an optimal threshold from the training data to map the relevance scores to [-1,1].

We provide the entire binary network as input and then attempt to predict the tie strength associated with each edge of the network. Therefore, we use root-mean-square error (RMSE) as the metric to evaluate the performance of tie strength prediction. 

\begin{table}
	\begin{center}
	\small
	\vskip -1em
	\caption{\label{tab:tie_strength}Performance comparison of tie strength prediction under the undirected setting.}
		\begin{tabular}{|c|c|c|}	\hline
			Metrics		&	Bitcoin-Alpha	&	Bitcoin-OTC	\\ \hline		
			UCN-R		&	0.317		&	0.286	 \\
			UCN-I		&	0.323		&	0.289	 \\
			SCN			&	\textbf{0.302}	&	\textbf{0.278}	 \\
			UJI-R		&	0.317		&	0.286	 \\
			UJI-I		&	0.323		&	0.289	 \\			
			SJI			&	\textbf{0.302}	&	\textbf{0.278}	 \\
			UPA-R		&	0.408		&	0.365	 \\
			UPA-I		&	0.429		&	0.370	 \\
			SPA			&	0.324			&	0.291	 \\ \hline
			UK-R			&	0.378		&	0.336	 \\
			UK-I		&	0.402		&	0.344	 \\
			SK			&	0.311	&	0.283	 \\ 
			
			URWR-R		&	0.410		&	0.379	 \\
			URWR-I		&	0.433		&	0.384	 \\
			SRWR			&	0.328		&	0.296	 \\
			UASCOS++-R	&	0.410		&	0.372	 \\			
			UASCOS++-I	&	0.438		&	0.378	 \\			
			SASCOS++		&	0.328		&	0.296	 \\ \hline
		\end{tabular}
	\end{center}
		
		\vskip 2em
		\begin{center}
	\small
	\vskip -1em
	\caption{\label{tab:tiestrength_directed}Performance comparison of tie-strength prediction under the directed setting.}
		\begin{tabular}{|c|c|c|}	\hline	
					Metrics		&	Bitcoin-Alpha	&	Bitcoin-OTC	 \\ \hline	
			UASCOS++-R	&	0.417	&	0.389	 \\
			UASCOS++-I	&	0.405	&	0.386	 \\			
			SASCOS++		&	\textbf{0.355}		&	\textbf{0.315}	 \\			
			URWR-R		&	0.403	&	0.373	 \\		
			URWR-F		&	0.416	&	0.392	 \\
			SRWR			&	0.362	&	0.325	 \\	 \hline
		\end{tabular}
	\end{center}
\end{table}

\subsubsection{Tie Strength Prediction Performance}

The tie strength prediction performance is demonstrated in Table~\ref{tab:tie_strength} and Table~\ref{tab:tiestrength_directed} for undirected and directed settings, respectively. It can be observed from the Table~\ref{tab:tie_strength} for the undirected setting : 
\begin{itemize}
\item Signed specific measurements remarkably outperform these by (1) removing negative links or (2) ignoring signs for tie strength prediction. This further supports the importance of negative links in signed relevance measurements. 
\item Local signed measurements obtain comparable or even better performance than global signed measurements in tie strength prediction. This observation is different from that for link prediction.  To achieve better link prediction performance, we only need to predict sign accurately. However, for tie strengthen prediction, in addition to signs of links, we also need to predict the relevance strength correctly.  Thus, local information could be good at predicting relevance strength. In fact, most existing tie strength prediction algorithms for unsigned networks only use local information~\cite{gilbert2009predicting,xiang2010modeling}. 
\end{itemize}

For the directed setting, we have similar observations for tie strength prediction to link prediction.

\section{Conclusion}

Node relevance measurements have been extensively studied for unsigned social networks. In recent years, signed network analysis has attracted increasing attention. However, as a fundamental task, node relevance measurements are rather limited. In this paper, we offer an initial and comprehensive study on signed relevance measurements. We build numerous local and global  measurements guided by signed network properties and balance theory.  We further study the impact of signed relevance measurements on two signed network analysis tasks, i.e., link prediction and tie strength prediction.  Experimental results demonstrate that (1) dedicated efforts are necessary to build signed relevance measurements with negative links; (2) global methods significantly outperform local methods for link prediction
; while local methods obtain comparable or even slightly better performance than global methods for tie strength prediction. 

We will further investigate the following directions. First, we would like to study other social theories for signed networks and build novel relevance measurements based on them. Second, we will study the impact of signed relevance measurements on more signed network analysis tasks such as node classification and node embedding. Finally since properties of negative links are different from these of positive links, we will study signed network modeling.

\bibliographystyle{ACM-Reference-Format}
\bibliography{JT} 


\begin{thebibliography}{48}


\ifx \showCODEN    \undefined \def \showCODEN     #1{\unskip}     \fi
\ifx \showDOI      \undefined \def \showDOI       #1{#1}\fi
\ifx \showISBNx    \undefined \def \showISBNx     #1{\unskip}     \fi
\ifx \showISBNxiii \undefined \def \showISBNxiii  #1{\unskip}     \fi
\ifx \showISSN     \undefined \def \showISSN      #1{\unskip}     \fi
\ifx \showLCCN     \undefined \def \showLCCN      #1{\unskip}     \fi
\ifx \shownote     \undefined \def \shownote      #1{#1}          \fi
\ifx \showarticletitle \undefined \def \showarticletitle #1{#1}   \fi
\ifx \showURL      \undefined \def \showURL       {\relax}        \fi
\providecommand\bibfield[2]{#2}
\providecommand\bibinfo[2]{#2}
\providecommand\natexlab[1]{#1}
\providecommand\showeprint[2][]{arXiv:#2}

\bibitem[\protect\citeauthoryear{Adamic and Adar}{Adamic and Adar}{2003}]%
        {adamic2003friends}
\bibfield{author}{\bibinfo{person}{Lada~A Adamic} {and} \bibinfo{person}{Eytan
  Adar}.} \bibinfo{year}{2003}\natexlab{}.
\newblock \showarticletitle{Friends and neighbors on the web}.
\newblock \bibinfo{journal}{{\em Social networks\/}} \bibinfo{volume}{25},
  \bibinfo{number}{3} (\bibinfo{year}{2003}), \bibinfo{pages}{211--230}.
\newblock


\bibitem[\protect\citeauthoryear{Anchuri and Magdon-Ismail}{Anchuri and
  Magdon-Ismail}{2012}]%
        {anchuri2012communities}
\bibfield{author}{\bibinfo{person}{Pranay Anchuri} {and} \bibinfo{person}{Malik
  Magdon-Ismail}.} \bibinfo{year}{2012}\natexlab{}.
\newblock \showarticletitle{Communities and balance in signed networks: A
  spectral approach}. In \bibinfo{booktitle}{{\em Advances in Social Networks
  Analysis and Mining (ASONAM), 2012 IEEE/ACM International Conference on}}.
  IEEE, \bibinfo{pages}{235--242}.
\newblock


\bibitem[\protect\citeauthoryear{Backstrom and Leskovec}{Backstrom and
  Leskovec}{2011}]%
        {backstrom2011supervised}
\bibfield{author}{\bibinfo{person}{Lars Backstrom} {and} \bibinfo{person}{Jure
  Leskovec}.} \bibinfo{year}{2011}\natexlab{}.
\newblock \showarticletitle{Supervised random walks: predicting and
  recommending links in social networks}. In \bibinfo{booktitle}{{\em
  Proceedings of the fourth ACM international conference on Web search and data
  mining}}. ACM, \bibinfo{pages}{635--644}.
\newblock


\bibitem[\protect\citeauthoryear{Barab{\'a}si and Albert}{Barab{\'a}si and
  Albert}{1999a}]%
        {barabasi1999emergence}
\bibfield{author}{\bibinfo{person}{Albert-L{\'a}szl{\'o} Barab{\'a}si} {and}
  \bibinfo{person}{R{\'e}ka Albert}.} \bibinfo{year}{1999}\natexlab{a}.
\newblock \showarticletitle{Emergence of scaling in random networks}.
\newblock \bibinfo{journal}{{\em science\/}} \bibinfo{volume}{286},
  \bibinfo{number}{5439} (\bibinfo{year}{1999}), \bibinfo{pages}{509--512}.
\newblock


\bibitem[\protect\citeauthoryear{Barab{\'a}si and Albert}{Barab{\'a}si and
  Albert}{1999b}]%
        {barabasi1999pa}
\bibfield{author}{\bibinfo{person}{Albert-L{\'a}szl{\'o} Barab{\'a}si} {and}
  \bibinfo{person}{R{\'e}ka Albert}.} \bibinfo{year}{1999}\natexlab{b}.
\newblock \showarticletitle{Emergence of scaling in random networks}.
\newblock \bibinfo{journal}{{\em science\/}} \bibinfo{volume}{286},
  \bibinfo{number}{5439} (\bibinfo{year}{1999}), \bibinfo{pages}{509--512}.
\newblock


\bibitem[\protect\citeauthoryear{Bhagat, Cormode, and Muthukrishnan}{Bhagat
  et~al\mbox{.}}{2011}]%
        {bhagat2011node}
\bibfield{author}{\bibinfo{person}{Smriti Bhagat}, \bibinfo{person}{Graham
  Cormode}, {and} \bibinfo{person}{S Muthukrishnan}.}
  \bibinfo{year}{2011}\natexlab{}.
\newblock \showarticletitle{Node classification in social networks}.
\newblock In \bibinfo{booktitle}{{\em Social network data analytics}}.
  \bibinfo{publisher}{Springer}, \bibinfo{pages}{115--148}.
\newblock


\bibitem[\protect\citeauthoryear{Cartwright and Harary}{Cartwright and
  Harary}{1956}]%
        {cartwright1956structural}
\bibfield{author}{\bibinfo{person}{Dorwin Cartwright} {and}
  \bibinfo{person}{Frank Harary}.} \bibinfo{year}{1956}\natexlab{}.
\newblock \showarticletitle{Structural balance: a generalization of Heider's
  theory.}
\newblock \bibinfo{journal}{{\em Psychological review\/}} \bibinfo{volume}{63},
  \bibinfo{number}{5} (\bibinfo{year}{1956}), \bibinfo{pages}{277}.
\newblock


\bibitem[\protect\citeauthoryear{Chen and Giles}{Chen and Giles}{2013}]%
        {chen2013ascos}
\bibfield{author}{\bibinfo{person}{Hung-Hsuan Chen} {and}
  \bibinfo{person}{C~Lee Giles}.} \bibinfo{year}{2013}\natexlab{}.
\newblock \showarticletitle{ASCOS: an asymmetric network structure context
  similarity measure}. In \bibinfo{booktitle}{{\em Advances in Social Networks
  Analysis and Mining (ASONAM), 2013 IEEE/ACM International Conference on}}.
  IEEE, \bibinfo{pages}{442--449}.
\newblock


\bibitem[\protect\citeauthoryear{Chen and Giles}{Chen and Giles}{2015}]%
        {chen2015ascos}
\bibfield{author}{\bibinfo{person}{Hung-Hsuan Chen} {and}
  \bibinfo{person}{C~Lee Giles}.} \bibinfo{year}{2015}\natexlab{}.
\newblock \showarticletitle{Ascos++: An asymmetric similarity measure for
  weighted networks to address the problem of simrank}.
\newblock \bibinfo{journal}{{\em ACM Transactions on Knowledge Discovery from
  Data (TKDD)\/}} \bibinfo{volume}{10}, \bibinfo{number}{2}
  (\bibinfo{year}{2015}), \bibinfo{pages}{15}.
\newblock


\bibitem[\protect\citeauthoryear{Chiang, Natarajan, Tewari, and Dhillon}{Chiang
  et~al\mbox{.}}{2011}]%
        {chiang2011exploiting}
\bibfield{author}{\bibinfo{person}{Kai-Yang Chiang}, \bibinfo{person}{Nagarajan
  Natarajan}, \bibinfo{person}{Ambuj Tewari}, {and} \bibinfo{person}{Inderjit~S
  Dhillon}.} \bibinfo{year}{2011}\natexlab{}.
\newblock \showarticletitle{Exploiting longer cycles for link prediction in
  signed networks}. In \bibinfo{booktitle}{{\em Proceedings of the 20th ACM
  international conference on Information and knowledge management}}. ACM,
  \bibinfo{pages}{1157--1162}.
\newblock


\bibitem[\protect\citeauthoryear{Chiang, Whang, and Dhillon}{Chiang
  et~al\mbox{.}}{2012}]%
        {chiang2012scalable}
\bibfield{author}{\bibinfo{person}{Kai-Yang Chiang},
  \bibinfo{person}{Joyce~Jiyoung Whang}, {and} \bibinfo{person}{Inderjit~S
  Dhillon}.} \bibinfo{year}{2012}\natexlab{}.
\newblock \showarticletitle{Scalable clustering of signed networks using
  balance normalized cut}. In \bibinfo{booktitle}{{\em Proceedings of the 21st
  ACM international conference on Information and knowledge management}}. ACM,
  \bibinfo{pages}{615--624}.
\newblock


\bibitem[\protect\citeauthoryear{Facchetti, Iacono, and Altafini}{Facchetti
  et~al\mbox{.}}{2011}]%
        {facchetti2011computing}
\bibfield{author}{\bibinfo{person}{Giuseppe Facchetti},
  \bibinfo{person}{Giovanni Iacono}, {and} \bibinfo{person}{Claudio Altafini}.}
  \bibinfo{year}{2011}\natexlab{}.
\newblock \showarticletitle{Computing global structural balance in large-scale
  signed social networks}.
\newblock \bibinfo{journal}{{\em Proceedings of the National Academy of
  Sciences\/}} \bibinfo{volume}{108}, \bibinfo{number}{52}
  (\bibinfo{year}{2011}), \bibinfo{pages}{20953--20958}.
\newblock


\bibitem[\protect\citeauthoryear{Gilbert and Karahalios}{Gilbert and
  Karahalios}{2009}]%
        {gilbert2009predicting}
\bibfield{author}{\bibinfo{person}{Eric Gilbert} {and} \bibinfo{person}{Karrie
  Karahalios}.} \bibinfo{year}{2009}\natexlab{}.
\newblock \showarticletitle{Predicting tie strength with social media}. In
  \bibinfo{booktitle}{{\em Proceedings of the SIGCHI conference on human
  factors in computing systems}}. ACM, \bibinfo{pages}{211--220}.
\newblock


\bibitem[\protect\citeauthoryear{Guha, Kumar, Raghavan, and Tomkins}{Guha
  et~al\mbox{.}}{2004}]%
        {guha2004propagation}
\bibfield{author}{\bibinfo{person}{Ramanthan Guha}, \bibinfo{person}{Ravi
  Kumar}, \bibinfo{person}{Prabhakar Raghavan}, {and} \bibinfo{person}{Andrew
  Tomkins}.} \bibinfo{year}{2004}\natexlab{}.
\newblock \showarticletitle{Propagation of trust and distrust}. In
  \bibinfo{booktitle}{{\em Proceedings of the 13th international conference on
  World Wide Web}}. ACM, \bibinfo{pages}{403--412}.
\newblock


\bibitem[\protect\citeauthoryear{Heider}{Heider}{1946}]%
        {heider1946attitudes}
\bibfield{author}{\bibinfo{person}{Fritz Heider}.}
  \bibinfo{year}{1946}\natexlab{}.
\newblock \showarticletitle{Attitudes and cognitive organization}.
\newblock \bibinfo{journal}{{\em The Journal of psychology\/}}
  \bibinfo{volume}{21}, \bibinfo{number}{1} (\bibinfo{year}{1946}),
  \bibinfo{pages}{107--112}.
\newblock


\bibitem[\protect\citeauthoryear{Horn and Johnson}{Horn and Johnson}{2013}]%
        {horn2012matrix}
\bibfield{author}{\bibinfo{person}{Roger~A. Horn} {and}
  \bibinfo{person}{Charles~R. Johnson}.} \bibinfo{year}{2013}\natexlab{}.
\newblock \bibinfo{booktitle}{{\em Matrix analysis\/}
  (\bibinfo{edition}{second} ed.)}.
\newblock \bibinfo{publisher}{Cambridge University Press, Cambridge}. xviii+643
  pages.
\newblock
\showISBNx{978-0-521-54823-6}


\bibitem[\protect\citeauthoryear{Hsieh, Chiang, and Dhillon}{Hsieh
  et~al\mbox{.}}{2012}]%
        {hsieh2012low}
\bibfield{author}{\bibinfo{person}{Cho-Jui Hsieh}, \bibinfo{person}{Kai-Yang
  Chiang}, {and} \bibinfo{person}{Inderjit~S Dhillon}.}
  \bibinfo{year}{2012}\natexlab{}.
\newblock \showarticletitle{Low rank modeling of signed networks}. In
  \bibinfo{booktitle}{{\em Proceedings of the 18th ACM SIGKDD international
  conference on Knowledge discovery and data mining}}. ACM,
  \bibinfo{pages}{507--515}.
\newblock


\bibitem[\protect\citeauthoryear{Jeh and Widom}{Jeh and Widom}{2002}]%
        {jeh2002simrank}
\bibfield{author}{\bibinfo{person}{Glen Jeh} {and} \bibinfo{person}{Jennifer
  Widom}.} \bibinfo{year}{2002}\natexlab{}.
\newblock \showarticletitle{SimRank: a measure of structural-context
  similarity}. In \bibinfo{booktitle}{{\em Proceedings of the eighth ACM SIGKDD
  international conference on Knowledge discovery and data mining}}. ACM,
  \bibinfo{pages}{538--543}.
\newblock


\bibitem[\protect\citeauthoryear{Jung, Jin, Sael, and Kang}{Jung
  et~al\mbox{.}}{2016}]%
        {Jung2016srwr}
\bibfield{author}{\bibinfo{person}{Jinhong Jung}, \bibinfo{person}{Woojeong
  Jin}, \bibinfo{person}{Lee Sael}, {and} \bibinfo{person}{U. Kang}.}
  \bibinfo{year}{2016}\natexlab{}.
\newblock \showarticletitle{Personalized Ranking in Signed Networks Using
  Signed Random Walk with Restart}. In \bibinfo{booktitle}{{\em {IEEE} 16th
  International Conference on Data Mining, {ICDM} 2016, December 12-15, 2016,
  Barcelona, Spain}}. \bibinfo{pages}{973--978}.
\newblock
\showDOI{%
\url{https://doi.org/10.1109/ICDM.2016.0122}}


\bibitem[\protect\citeauthoryear{Kahanda and Neville}{Kahanda and
  Neville}{2009}]%
        {kahanda2009strength}
\bibfield{author}{\bibinfo{person}{Indika Kahanda} {and}
  \bibinfo{person}{Jennifer Neville}.} \bibinfo{year}{2009}\natexlab{}.
\newblock \showarticletitle{Using Transactional Information to Predict Link
  Strength in Online Social Networks}. In \bibinfo{booktitle}{{\em Third
  International AAAI Conference on Weblogs and Social Media}}.
\newblock


\bibitem[\protect\citeauthoryear{Katz}{Katz}{1953}]%
        {katz1953sim}
\bibfield{author}{\bibinfo{person}{Leo Katz}.} \bibinfo{year}{1953}\natexlab{}.
\newblock \showarticletitle{A new status index derived from sociometric
  analysis}.
\newblock \bibinfo{journal}{{\em Psychometrika\/}} \bibinfo{volume}{18},
  \bibinfo{number}{1} (\bibinfo{year}{1953}), \bibinfo{pages}{39--43}.
\newblock


\bibitem[\protect\citeauthoryear{Kunegis, Lommatzsch, and Bauckhage}{Kunegis
  et~al\mbox{.}}{2009}]%
        {kunegis2009slashdot}
\bibfield{author}{\bibinfo{person}{J{\'e}r{\^o}me Kunegis},
  \bibinfo{person}{Andreas Lommatzsch}, {and} \bibinfo{person}{Christian
  Bauckhage}.} \bibinfo{year}{2009}\natexlab{}.
\newblock \showarticletitle{The slashdot zoo: mining a social network with
  negative edges}. In \bibinfo{booktitle}{{\em Proceedings of the 18th
  international conference on World wide web}}. ACM, \bibinfo{pages}{741--750}.
\newblock


\bibitem[\protect\citeauthoryear{Kunegis, Schmidt, Lommatzsch, Lerner, De~Luca,
  and Albayrak}{Kunegis et~al\mbox{.}}{2010}]%
        {kunegis2010spectral}
\bibfield{author}{\bibinfo{person}{J{\'e}r{\^o}me Kunegis},
  \bibinfo{person}{Stephan Schmidt}, \bibinfo{person}{Andreas Lommatzsch},
  \bibinfo{person}{J{\"u}rgen Lerner}, \bibinfo{person}{Ernesto~W De~Luca},
  {and} \bibinfo{person}{Sahin Albayrak}.} \bibinfo{year}{2010}\natexlab{}.
\newblock \showarticletitle{Spectral analysis of signed graphs for clustering,
  prediction and visualization}. In \bibinfo{booktitle}{{\em Proceedings of the
  2010 SIAM International Conference on Data Mining}}. SIAM,
  \bibinfo{pages}{559--570}.
\newblock


\bibitem[\protect\citeauthoryear{Leskovec, Huttenlocher, and
  Kleinberg}{Leskovec et~al\mbox{.}}{2010a}]%
        {leskovec2010predicting}
\bibfield{author}{\bibinfo{person}{Jure Leskovec}, \bibinfo{person}{Daniel
  Huttenlocher}, {and} \bibinfo{person}{Jon Kleinberg}.}
  \bibinfo{year}{2010}\natexlab{a}.
\newblock \showarticletitle{Predicting positive and negative links in online
  social networks}. In \bibinfo{booktitle}{{\em Proceedings of the 19th
  international conference on World wide web}}. ACM, \bibinfo{pages}{641--650}.
\newblock


\bibitem[\protect\citeauthoryear{Leskovec, Huttenlocher, and
  Kleinberg}{Leskovec et~al\mbox{.}}{2010b}]%
        {leskovec2010signed}
\bibfield{author}{\bibinfo{person}{Jure Leskovec}, \bibinfo{person}{Daniel
  Huttenlocher}, {and} \bibinfo{person}{Jon Kleinberg}.}
  \bibinfo{year}{2010}\natexlab{b}.
\newblock \showarticletitle{Signed networks in social media}. In
  \bibinfo{booktitle}{{\em Proceedings of the SIGCHI conference on human
  factors in computing systems}}. ACM, \bibinfo{pages}{1361--1370}.
\newblock


\bibitem[\protect\citeauthoryear{Li, Tang, Wang, Wan, Chang, and Liu}{Li
  et~al\mbox{.}}{2017}]%
        {li2017understanding}
\bibfield{author}{\bibinfo{person}{Jundong Li}, \bibinfo{person}{Jiliang Tang},
  \bibinfo{person}{Yilin Wang}, \bibinfo{person}{Yali Wan}, \bibinfo{person}{Yi
  Chang}, {and} \bibinfo{person}{Huan Liu}.} \bibinfo{year}{2017}\natexlab{}.
\newblock \showarticletitle{Understanding and Predicting Delay in Reciprocal
  Relations}.
\newblock \bibinfo{journal}{{\em arXiv preprint arXiv:1703.01393\/}}
  (\bibinfo{year}{2017}).
\newblock


\bibitem[\protect\citeauthoryear{Liben-Nowell and Kleinberg}{Liben-Nowell and
  Kleinberg}{2007}]%
        {liben2007link}
\bibfield{author}{\bibinfo{person}{David Liben-Nowell} {and}
  \bibinfo{person}{Jon Kleinberg}.} \bibinfo{year}{2007}\natexlab{}.
\newblock \showarticletitle{The link-prediction problem for social networks}.
\newblock \bibinfo{journal}{{\em journal of the Association for Information
  Science and Technology\/}} \bibinfo{volume}{58}, \bibinfo{number}{7}
  (\bibinfo{year}{2007}), \bibinfo{pages}{1019--1031}.
\newblock


\bibitem[\protect\citeauthoryear{Lorrain and White}{Lorrain and White}{1971}]%
        {lorrain1971structural}
\bibfield{author}{\bibinfo{person}{Francois Lorrain} {and}
  \bibinfo{person}{Harrison~C White}.} \bibinfo{year}{1971}\natexlab{}.
\newblock \showarticletitle{Structural equivalence of individuals in social
  networks}.
\newblock \bibinfo{journal}{{\em The Journal of mathematical sociology\/}}
  \bibinfo{volume}{1}, \bibinfo{number}{1} (\bibinfo{year}{1971}),
  \bibinfo{pages}{49--80}.
\newblock


\bibitem[\protect\citeauthoryear{L{\"u} and Zhou}{L{\"u} and Zhou}{2011}]%
        {lu2011link}
\bibfield{author}{\bibinfo{person}{Linyuan L{\"u}} {and} \bibinfo{person}{Tao
  Zhou}.} \bibinfo{year}{2011}\natexlab{}.
\newblock \showarticletitle{Link prediction in complex networks: A survey}.
\newblock \bibinfo{journal}{{\em Physica A: Statistical Mechanics and its
  Applications\/}} \bibinfo{volume}{390}, \bibinfo{number}{6}
  (\bibinfo{year}{2011}), \bibinfo{pages}{1150--1170}.
\newblock


\bibitem[\protect\citeauthoryear{Ma, Lyu, and King}{Ma et~al\mbox{.}}{2009}]%
        {ma2009learning}
\bibfield{author}{\bibinfo{person}{Hao Ma}, \bibinfo{person}{Michael~R Lyu},
  {and} \bibinfo{person}{Irwin King}.} \bibinfo{year}{2009}\natexlab{}.
\newblock \showarticletitle{Learning to recommend with trust and distrust
  relationships}. In \bibinfo{booktitle}{{\em Proceedings of the third ACM
  conference on Recommender systems}}. ACM, \bibinfo{pages}{189--196}.
\newblock


\bibitem[\protect\citeauthoryear{McPherson, Smith-Lovin, and Cook}{McPherson
  et~al\mbox{.}}{2001}]%
        {mcpherson2001birds}
\bibfield{author}{\bibinfo{person}{Miller McPherson}, \bibinfo{person}{Lynn
  Smith-Lovin}, {and} \bibinfo{person}{James~M Cook}.}
  \bibinfo{year}{2001}\natexlab{}.
\newblock \showarticletitle{Birds of a feather: Homophily in social networks}.
\newblock \bibinfo{journal}{{\em Annual review of sociology\/}}
  \bibinfo{volume}{27}, \bibinfo{number}{1} (\bibinfo{year}{2001}),
  \bibinfo{pages}{415--444}.
\newblock


\bibitem[\protect\citeauthoryear{Newman}{Newman}{2010}]%
        {newman2010intro}
\bibfield{author}{\bibinfo{person}{Mark Newman}.}
  \bibinfo{year}{2010}\natexlab{}.
\newblock \bibinfo{booktitle}{{\em Networks: An Introduction}}.
\newblock \bibinfo{publisher}{Oxford University Press, Inc.},
  \bibinfo{address}{New York, NY, USA}.
\newblock
\showISBNx{0199206651, 9780199206650}


\bibitem[\protect\citeauthoryear{Scott}{Scott}{2012}]%
        {scott2012social}
\bibfield{author}{\bibinfo{person}{John Scott}.}
  \bibinfo{year}{2012}\natexlab{}.
\newblock \bibinfo{booktitle}{{\em Social network analysis}}.
\newblock \bibinfo{publisher}{Sage}.
\newblock


\bibitem[\protect\citeauthoryear{Symeonidis and Tiakas}{Symeonidis and
  Tiakas}{2014}]%
        {symeonidis2014transitive}
\bibfield{author}{\bibinfo{person}{Panagiotis Symeonidis} {and}
  \bibinfo{person}{Eleftherios Tiakas}.} \bibinfo{year}{2014}\natexlab{}.
\newblock \showarticletitle{Transitive node similarity: predicting and
  recommending links in signed social networks}.
\newblock \bibinfo{journal}{{\em World Wide Web\/}} \bibinfo{volume}{17},
  \bibinfo{number}{4} (\bibinfo{year}{2014}), \bibinfo{pages}{743--776}.
\newblock


\bibitem[\protect\citeauthoryear{Szell, Lambiotte, and Thurner}{Szell
  et~al\mbox{.}}{2010}]%
        {szell2010multirelational}
\bibfield{author}{\bibinfo{person}{Michael Szell}, \bibinfo{person}{Renaud
  Lambiotte}, {and} \bibinfo{person}{Stefan Thurner}.}
  \bibinfo{year}{2010}\natexlab{}.
\newblock \showarticletitle{Multirelational organization of large-scale social
  networks in an online world}.
\newblock \bibinfo{journal}{{\em Proceedings of the National Academy of
  Sciences\/}} \bibinfo{volume}{107}, \bibinfo{number}{31}
  (\bibinfo{year}{2010}), \bibinfo{pages}{13636--13641}.
\newblock


\bibitem[\protect\citeauthoryear{Tang, Chang, Aggarwal, and Liu}{Tang
  et~al\mbox{.}}{2015}]%
        {tang2015negative}
\bibfield{author}{\bibinfo{person}{Jiliang Tang}, \bibinfo{person}{Shiyu
  Chang}, \bibinfo{person}{Charu Aggarwal}, {and} \bibinfo{person}{Huan Liu}.}
  \bibinfo{year}{2015}\natexlab{}.
\newblock \showarticletitle{Negative link prediction in social media}. In
  \bibinfo{booktitle}{{\em Proceedings of the Eighth ACM International
  Conference on Web Search and Data Mining}}. ACM, \bibinfo{pages}{87--96}.
\newblock


\bibitem[\protect\citeauthoryear{Tang, Chang, Aggarwal, and Liu}{Tang
  et~al\mbox{.}}{2016}]%
        {tang2016survey}
\bibfield{author}{\bibinfo{person}{Jiliang Tang}, \bibinfo{person}{Yi Chang},
  \bibinfo{person}{Charu Aggarwal}, {and} \bibinfo{person}{Huan Liu}.}
  \bibinfo{year}{2016}\natexlab{}.
\newblock \showarticletitle{A survey of signed network mining in social media}.
\newblock \bibinfo{journal}{{\em ACM Computing Surveys (CSUR)\/}}
  \bibinfo{volume}{49}, \bibinfo{number}{3} (\bibinfo{year}{2016}),
  \bibinfo{pages}{42}.
\newblock


\bibitem[\protect\citeauthoryear{Tang, Hu, and Liu}{Tang et~al\mbox{.}}{2014}]%
        {tang2014distrust}
\bibfield{author}{\bibinfo{person}{Jiliang Tang}, \bibinfo{person}{Xia Hu},
  {and} \bibinfo{person}{Huan Liu}.} \bibinfo{year}{2014}\natexlab{}.
\newblock \showarticletitle{Is distrust the negation of trust?: the value of
  distrust in social media}. In \bibinfo{booktitle}{{\em Proceedings of the
  25th ACM conference on Hypertext and social media}}. ACM,
  \bibinfo{pages}{148--157}.
\newblock


\bibitem[\protect\citeauthoryear{Tang and Liu}{Tang and Liu}{2010}]%
        {tang2010community}
\bibfield{author}{\bibinfo{person}{Lei Tang} {and} \bibinfo{person}{Huan Liu}.}
  \bibinfo{year}{2010}\natexlab{}.
\newblock \showarticletitle{Community detection and mining in social media}.
\newblock \bibinfo{journal}{{\em Synthesis Lectures on Data Mining and
  Knowledge Discovery\/}} \bibinfo{volume}{2}, \bibinfo{number}{1}
  (\bibinfo{year}{2010}), \bibinfo{pages}{1--137}.
\newblock


\bibitem[\protect\citeauthoryear{Tong, Faloutsos, and Pan}{Tong
  et~al\mbox{.}}{2006}]%
        {tong2006fast}
\bibfield{author}{\bibinfo{person}{Hanghang Tong}, \bibinfo{person}{Christos
  Faloutsos}, {and} \bibinfo{person}{Jia-yu Pan}.}
  \bibinfo{year}{2006}\natexlab{}.
\newblock \showarticletitle{Fast Random Walk with Restart and Its
  Applications}. In \bibinfo{booktitle}{{\em Data Mining, 2006. ICDM'06. Sixth
  International Conference on}}. IEEE, \bibinfo{pages}{613--622}.
\newblock


\bibitem[\protect\citeauthoryear{Victor, Cornelis, De~Cock, and
  Teredesai}{Victor et~al\mbox{.}}{2009}]%
        {victor2009trust}
\bibfield{author}{\bibinfo{person}{Patricia Victor}, \bibinfo{person}{Chris
  Cornelis}, \bibinfo{person}{Martine De~Cock}, {and} \bibinfo{person}{Ankur
  Teredesai}.} \bibinfo{year}{2009}\natexlab{}.
\newblock \showarticletitle{Trust-and distrust-based recommendations for
  controversial reviews}. In \bibinfo{booktitle}{{\em Web Science Conference
  (WebSci'09: Society On-Line)}}.
\newblock


\bibitem[\protect\citeauthoryear{Wang, Tang, Aggarwal, Chang, and Liu}{Wang
  et~al\mbox{.}}{2017}]%
        {wang2017signed}
\bibfield{author}{\bibinfo{person}{Suhang Wang}, \bibinfo{person}{Jiliang
  Tang}, \bibinfo{person}{Charu Aggarwal}, \bibinfo{person}{Yi Chang}, {and}
  \bibinfo{person}{Huan Liu}.} \bibinfo{year}{2017}\natexlab{}.
\newblock \showarticletitle{Signed network embedding in social media}. SDM.
\newblock


\bibitem[\protect\citeauthoryear{Wasserman and Faust}{Wasserman and
  Faust}{1994}]%
        {wasserman1994social}
\bibfield{author}{\bibinfo{person}{Stanley Wasserman} {and}
  \bibinfo{person}{Katherine Faust}.} \bibinfo{year}{1994}\natexlab{}.
\newblock \bibinfo{booktitle}{{\em Social network analysis: Methods and
  applications}}. Vol.~\bibinfo{volume}{8}.
\newblock \bibinfo{publisher}{Cambridge university press}.
\newblock


\bibitem[\protect\citeauthoryear{Xiang, Neville, and Rogati}{Xiang
  et~al\mbox{.}}{2010}]%
        {xiang2010modeling}
\bibfield{author}{\bibinfo{person}{Rongjing Xiang}, \bibinfo{person}{Jennifer
  Neville}, {and} \bibinfo{person}{Monica Rogati}.}
  \bibinfo{year}{2010}\natexlab{}.
\newblock \showarticletitle{Modeling relationship strength in online social
  networks}. In \bibinfo{booktitle}{{\em Proceedings of the 19th international
  conference on World wide web}}. ACM, \bibinfo{pages}{981--990}.
\newblock


\bibitem[\protect\citeauthoryear{Yang, Cheung, and Liu}{Yang
  et~al\mbox{.}}{2007}]%
        {yang2007community}
\bibfield{author}{\bibinfo{person}{Bo Yang}, \bibinfo{person}{William Cheung},
  {and} \bibinfo{person}{Jiming Liu}.} \bibinfo{year}{2007}\natexlab{}.
\newblock \showarticletitle{Community mining from signed social networks}.
\newblock \bibinfo{journal}{{\em IEEE transactions on knowledge and data
  engineering\/}} \bibinfo{volume}{19}, \bibinfo{number}{10}
  (\bibinfo{year}{2007}).
\newblock


\bibitem[\protect\citeauthoryear{Yin, Cui, Li, Yao, and Chen}{Yin
  et~al\mbox{.}}{2012}]%
        {yin2012longtail}
\bibfield{author}{\bibinfo{person}{Hongzhi Yin}, \bibinfo{person}{Bin Cui},
  \bibinfo{person}{Jing Li}, \bibinfo{person}{Junjie Yao}, {and}
  \bibinfo{person}{Chen Chen}.} \bibinfo{year}{2012}\natexlab{}.
\newblock \showarticletitle{Challenging the Long Tail Recommendation}.
\newblock \bibinfo{journal}{{\em Proc. VLDB Endow.\/}} \bibinfo{volume}{5},
  \bibinfo{number}{9} (\bibinfo{date}{May} \bibinfo{year}{2012}),
  \bibinfo{pages}{896--907}.
\newblock
\showISSN{2150-8097}
\showDOI{%
\url{https://doi.org/10.14778/2311906.2311916}}


\bibitem[\protect\citeauthoryear{Yin, Gupta, Weninger, and Han}{Yin
  et~al\mbox{.}}{2010}]%
        {yin2010unified}
\bibfield{author}{\bibinfo{person}{Zhijun Yin}, \bibinfo{person}{Manish Gupta},
  \bibinfo{person}{Tim Weninger}, {and} \bibinfo{person}{Jiawei Han}.}
  \bibinfo{year}{2010}\natexlab{}.
\newblock \showarticletitle{A unified framework for link recommendation using
  random walks}. In \bibinfo{booktitle}{{\em Advances in Social Networks
  Analysis and Mining (ASONAM), 2010 International Conference on}}. IEEE,
  \bibinfo{pages}{152--159}.
\newblock


\bibitem[\protect\citeauthoryear{Zheng, Zeng, and Wang}{Zheng
  et~al\mbox{.}}{2015}]%
        {zheng2015social}
\bibfield{author}{\bibinfo{person}{Xiaolong Zheng}, \bibinfo{person}{Daniel
  Zeng}, {and} \bibinfo{person}{Fei-Yue Wang}.}
  \bibinfo{year}{2015}\natexlab{}.
\newblock \showarticletitle{Social balance in signed networks}.
\newblock \bibinfo{journal}{{\em Information Systems Frontiers\/}}
  \bibinfo{volume}{17}, \bibinfo{number}{5} (\bibinfo{year}{2015}),
  \bibinfo{pages}{1077--1095}.
\newblock


\end{thebibliography}

\end{document}